\documentclass[envcountsame]{llncs}

\usepackage[caption=false]{subfig}

\usepackage{graphicx}

\usepackage{amsmath}

\usepackage{amsfonts}
\usepackage{amssymb}
\usepackage{amsthm}
\usepackage{mathtools}

\usepackage{pgfplots}
\pgfplotsset{compat=1.14}

\usepackage[linesnumbered,boxruled]{algorithm2e}

\SetArgSty{normalfont}

\SetKwInOut{Input}{Input}
\SetKwInOut{Output}{Output}
\SetKwInOut{Effect}{Effect}

\SetKw{Return}{return}
\SetKw{Error}{error}
\newcommand{\algoIfrule}{\SetKwIF{Ifrule}{ElseIfrule}{Else}{if}{then}{else if}{else}{endif}}

\DeclareMathOperator{\RPiv}{\normalfont{\texttt{PivR}}}
\DeclareMathOperator{\IPiv}{\normalfont{\texttt{PivI}}}
\DeclareMathOperator{\ExtendRat}{\normalfont{\texttt{ExtendR}}}
\DeclareMathOperator{\ExtendInt}{\normalfont{\texttt{ExtendI}}}
\DeclareMathOperator{\ExtendMEH}{\normalfont{\texttt{ExtendMEH}}}
\DeclareMathOperator{\ReduceLeftInt}{\normalfont{\texttt{ReduceLeftI}}}
\DeclareMathOperator{\ReduceRightInt}{\normalfont{\texttt{ReduceRightI}}}
\DeclareMathOperator{\AbstractToInt}{\normalfont{\texttt{AbstractToInt}}}

\newcommand{\mop}[1]{\textrm{#1}}

\newcommand{\lcm}{\mop{lcm}}
\newcommand{\floor}{\mop{floor}}

\newcommand\piv[2]{\mop{piv}(#1,#2)}

\newcommand\Pset[2]{\mathcal{Q}(#1 x \leq #2)}
\newcommand\Rset[1]{\mathcal{Q}(#1)}

\newcommand\Mset[1]{\mathcal{M}(#1)}

\newcommand\centermath[1]{\newline \centerline{$#1$}}

\newcommand\centerequation[1]{\refstepcounter{equation}\newline \centerline{\hfill $#1$ \hfill (\theequation)}}

\theoremstyle{definition}

\title{A Reduction from\\Unbounded Linear Mixed Arithmetic Problems\\into Bounded Problems}
\author{Martin Bromberger\inst{1,2}}
\institute{Max Planck Institute for Informatics and Saarland University, Saarland Informatics Campus, Germany   \email{mbromber@mpi-inf.mpg.de} \and
           Graduate School of Computer Science, Saarland Informatics Campus, Germany}
\titlerunning{Bounding Transformations for LA}

\begin{document}
\pagestyle{plain}

\maketitle

\begin{abstract}
  We present a combination of the Mixed-Echelon-Hermite transformation
  and the Double-Bounded Reduction for systems of linear mixed arithmetic that preserve
  satisfiability and can be computed in polynomial time.
  Together, the two transformations turn any system of linear mixed constraints into a bounded system, 
  i.e., a system for which termination can be achieved easily. 
  Existing approaches for linear mixed arithmetic, 
  e.g., branch-and-bound and cuts from proofs, only explore
  a finite search space after application of our two transformations.
  Instead of generating \textit{a priori} bounds for the variables, e.g., as suggested
  by Papadimitriou, unbounded variables are eliminated through the two transformations.
  The transformations orient themselves on the structure of 
  an input system instead of computing \textit{a priori} (over-)approximations out of the available constants.
  Experiments provide further evidence to the efficiency of the transformations in practice.
  We also present a polynomial method for converting certificates of
  (un)satisfiability from the transformed to the original system.
    
\begin{keywords}
Linear Arithmetic, Integer Arithmetic, Mixed Arithmetic, SMT, Linear Transformations, Constraint Solving
\end{keywords}

\end{abstract}

\section{Introduction}
\label{SE: Introduction}
\label{SE: Introduction}


Efficient linear arithmetic decision procedures are important for various independent research lines, e.g., 
optimization, system modeling, and verification. 
We are interested in feasibility of linear arithmetic problems in the context of
the combination of theories, as they occur, e.g., in SMT solving or theorem proving. 


The SMT and theorem proving communities have presented several interesting and efficient approaches for pure linear rational arithmetic~\cite{DutertredeMoura:06} as well as linear integer arithmetic~\cite{BobotCCIMMM:12,BrombergerSturmWeidenbach:15,Dillig:09,Griggio:12}. 
SMT research also starts to extend into linear mixed arithmetic~\cite{ChristHoenicke2015,DutertredeMoura:06} 
because some applications require both rational and integer variables, 
e.g., planning/scheduling problems and verification of timed automata and hybrid systems. 

We are interest in decision procedures for mixed arithmetic because of a possible combination with superposition~\cite{AlthausKW:09,BaumgartnerWaldmann2013CADE,FietzkeWeidenbach12}. 
In the superposition context, arithmetic constraints are part of the first-order clauses. 
The problems are typically unbounded due to transformations that turn the input formula into a superposition specific input format. 
Since these problems are unbounded, the search space becomes infinite, 
which is the case where termination becomes difficult for most linear arithmetic approaches. 
Unbounded problems appear also in other areas of automated reasoning. 
Either because of bad encodings, necessary but complicating transformations, e.g., slacking (see Section~\ref{SE:Experiments}), 
or the sheer complexity of the verification goal. 
Hence, efficient techniques for handling unbounded problems are necessary for a generally reliable combined procedure.


It is theoretically very easy to achieve termination for linear integer and mixed arithmetic because 
of so called \emph{a priori bounds}. 
For example, the \textit{a priori} bounds presented by Papadimitriou~\cite{Papadimitriou:81} 
guarantee that a problem has a mixed solution if and only if the problem extended by the bounds $|x_i| \leq 2 n (m a)^{2m+1}$ for every variable $x_i$ has a mixed solution. 
In these \textit{a priori} bounds, $n$ is the number of variables, $m$ the number of inequalities, and $a$ the largest absolute value of any integer coefficient or constant in the problem. 
By extending a problem with those \textit{a priori} bounds, we reduce the search space for a branch-and-bound solver (and many other mixed arithmetic decision procedures) to a finite search space. 
So branch-and-bound is guaranteed to terminate. 

However, these bounds are so large that the resulting search space 
cannot be explored in reasonable time for many practical problems. 
One reason for the impracticability of \textit{a priori} bounds is that they only take parameter sizes into account and not actually the structure of each problem.
\textit{A priori} bounds are not integrated in any state-of-the-art SMT solvers~\cite{BarrettCDHJKRT:11,ChristHoenickeNutz2012,CimattiGriggio:13,deMouraBjorner:08,Dutertre:14} since they are no help in practice. 
As far as we know, none of the state-of-the-art SMT solvers use any method that guarantees termination for linear integer or mixed arithmetic. 

In this paper, 
we present satisfiability preserving transformations 
that reduce unbounded problems into bounded problems. 
On these bounded problems, most linear mixed decision procedures become terminating, which we show on the example of branch-and-bound. 
Our reduction works by eliminating unbounded variables. 
First, we use the Double-Bounded reduction (Section~\ref{SE:splitting}) to eliminate all unbounded inequalities from our constraint system. 
Then we use the Mixed-Echelon-Hermite transformation (Section~\ref{SE:doublebounded}) to shift the variables of our system to ones that are either bounded or do not appear in the new inequalities and are, therefore, eliminated. 
With Corollary~\ref{corollary:MCTCERT} \& Lemma~\ref{lemma:mixedsoundness} we explain how to efficiently convert certificates of (un)satisfiability between the transformed and the original system.
Our method is efficient because it is fully guided by the structure of the problem. 
This is confirmed by experiments (Section~\ref{SE:Experiments}).
We also show how to efficiently determine when a problem is unbounded (Lemma~\ref{lemma:boundsandequalities}). 
This prevents our solver from applying our transformations on bounded problems.

In Appendix~\ref{SE:incremental} of this paper, we explain how to implement the presented procedures in an incrementally efficient way. 
This is relevant for an efficient SMT implementation.

The original version of this paper has been accepted by IJCAR 2018 
and will be published by Springer as part of the Lecture Notes of Computer Science Series.

\section{Preliminaries}
\label{SE:Preliminaries}

While the difference between matrices, vectors, and their components is always clear in context, 
we generally use upper case letters for matrices (e.g., $A$), lower case letters for vectors (e.g., $x$), 
and lower case letters with an index $i$ or $j$ (e.g., $b_i$, $x_j$) as components of the associated vector at position $i$ or $j$, respectively.
The only exceptions are the row vectors $a_i^T = (a_{i1}, \ldots, a_{in})$ of a matrix $A = (a_1, \ldots, a_m)^T$, which already contain an index $i$ that indicates the row's position inside $A$.
We also abbreviate the $n$-dimensional origin $(0, \ldots, 0)^T$ as $0^n$.  
Moreover, we denote by $\piv{A}{j}$ the row index of the \emph{pivot} of a column $j$, i.e., the smallest row index $i$ with a non-zero entry $a_{ij}$ or $m+j$ if there are no non-zero entries in column $j$.

A system of constraints $A x \leq b$ is just a set of non-strict inequalities\footnote{All techniques discussed in this paper can be extended to strict inequalities with the help of $\delta$-rationals~\cite{DutertredeMoura:06}. We will omit the strict inequalities and focus only on non-strict inequalities due to lack of space.} $\{a_1^T x \leq b_1, \ldots, a_m^T x \leq b_m\}$ and the \emph{rational solutions} of this system are exactly those points $x \in \mathbb{Q}^n$ that satisfy all inequalities in this set. 
The row coefficients are given by $A = (a_1, \ldots, a_m)^T \in \mathbb{Q}^{m \times n}$, 
the variables are given by $x = (x_1, \ldots, x_n)^T$, and the inequality bounds are given by $b = (b_1, \ldots, b_m)^T \in \mathbb{Q}^{m}$. 
Moreover, we assume that any constant rows $a_i = 0^n$ were eliminated from our system during an implicit preprocessing step. 
This is a trivial task and eliminates some unnecessarily complicated corner cases. 

In this paper, we consider mixed constraint systems, i.e., variables are assigned a type: either rational or integer. 
Due to convenience, we assume that the first $n_1$ variables $(x_1, \ldots, x_{n_1})$ are rational and the remaining $n_2$ variables $(x_{n_1 + 1}, \ldots, x_n)$ are integer, where $n = n_1 + n_2$.
A \emph{mixed solution} is a point $x \in (\mathbb{Q}^{n_1} \times \mathbb{Z}^{n_2})$ that satisfy all inequalities in $A x \leq b$ and 
we denote by $\Mset{A x \leq b}= \{x \in (\mathbb{Q}^{n_1} \times \mathbb{Z}^{n_2}) : A x \leq b\}$ the \emph{set of mixed solutions} 
to the system of inequalities $A x \leq b$. 
We sometimes need to relax the variables to be completely rational. 
Therefore, we denote by $\Rset{A x \leq b}= \{x \in \mathbb{Q}^n : A x \leq b\}$ the \emph{set of rational solutions} 
to the system of inequalities $A x \leq b$.

Since $A x \leq b$ and $A' x \leq b'$ are just sets, we can write their combination as $(A x \leq b) \cup (A' x \leq b')$. 
A special system of inequalities is a system of equations $D x = c$, which is equivalent to the combined system of inequalities $(D x \leq c) \cup (-D x \leq -c)$. 
We say that a constraint system implies an inequality $h^T x \leq g$, 
where $h \in \mathbb{Q}^n$, $h \neq 0^n$, and $g \in \mathbb{Q}$, 
if $h^T x \leq g$ holds for all $x \in \Pset{A}{b}$.
In the same manner, a constraint system implies an equality $h^T x = g$, 
where $h \in \mathbb{Q}^n$, $h \neq 0^n$, and $g \in \mathbb{Q}$, 
if $h^T x = g$ holds for all $x \in \Pset{A}{b}$.
A constraint implied by $A x \leq b$ is \emph{explicit} if it does appear in $A x \leq b$. 
Otherwise, it is called \emph{implicit}.

Most deductions on linear inequalities are based on Farkas' Lemma: 

\begin{lemma}[Farkas' Lemma~\cite{Boyd:04}]
$\Rset{A x \leq b} = \emptyset$ iff there exists a $y \in \mathbb{Q}^m$ with $y \geq 0^m$ and $y^T A = 0^n$ so that $y^T b < 0$, i.e.,
there exists a non-negative linear combination of inequalities in $A x \leq b$ that results in an inequality $y^T A x \leq y^T b$ that is constant and unsatisfiable. If such a $y$ exists, then we call it a \emph{certificate of unsatisfiability}.
\label{lemma:farkasunsat}
\end{lemma}

We also frequently use the following lemma, which is just a reformulation of Farkas' Lemma:

\begin{lemma}[Linear Implication Lemma]
Let $\Rset{A x \leq b} \neq \emptyset$, $h \in \mathbb{Q}^n \setminus\{0^n\}$, and $g \in \mathbb{Q}$. Then, $A x \leq b$ implies $h^T x \leq g$ iff there exists a $y \in \mathbb{Q}^m$ with $y \geq 0^m$ and $y^T A = h^T$ so that $y^T b \leq  g$, i.e.,
there exists a non-negative linear combination of inequalities in $A x \leq b$ that results in the inequality $h^T x \leq g$.
\label{lemma:farkasimplies}
\end{lemma}
%

As we mentioned in the introduction, this paper describes equisatisfiable transformations for constraint systems. 
We transform the systems in such a way that most linear mixed decision procedures become terminating and still retain their general efficiency. 
We even show this on the example of branch-and-bound. 
Although we do not have the time to discuss all facets of branch-and-bound~\cite{Schrijver:86}, we still want to give a short summary of the algorithm. 
Branch-and-bound is a recursive algorithm that computes mixed solutions for constraint systems. 
In each call of the algorithm, it first computes a rational solution $s$ to a constraint system $A x \leq b$\footnote{A rational solution can be computed in polynomial time~\cite{Schrijver:86}.}. 
If there are none, then we know that $A x \leq b$ has no mixed solution. 
We are also done in the case that $s$ is a mixed solution. 
Otherwise, we select one of the integer variables $x_i$ assigned to a fractional value $s_i \not\in \mathbb{Z}$ and call branch-and-bound recursively on $(A x \leq b) \cup (x_i \geq \lceil s_i \rceil)$ and $(A x \leq b) \cup (x_i \leq \lfloor s_i \rfloor)$. 
If none of the recursive calls returns a mixed solution, then $A x \leq b$ also does not have a mixed solution. 
Likewise, if one of them returns a mixed solution $s$, then it also is a mixed solution to $A x \leq b$.

\begin{figure}[t]
    \subfloat[]{
	  \includegraphics[width=0.32\textwidth]{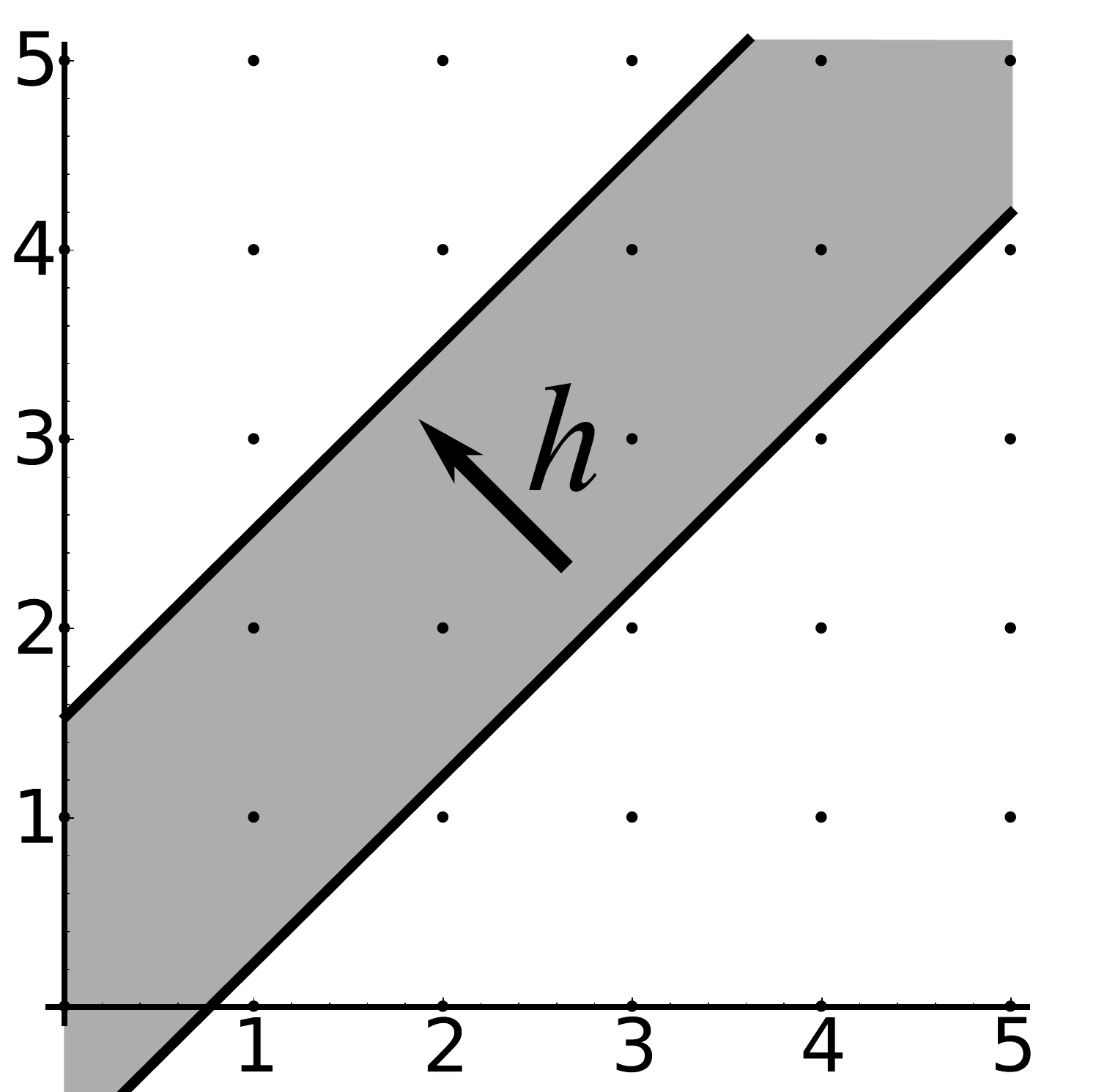}
      \label{fig:boundeddirection}
	}
	\subfloat[]{
      \includegraphics[width=0.32\textwidth]{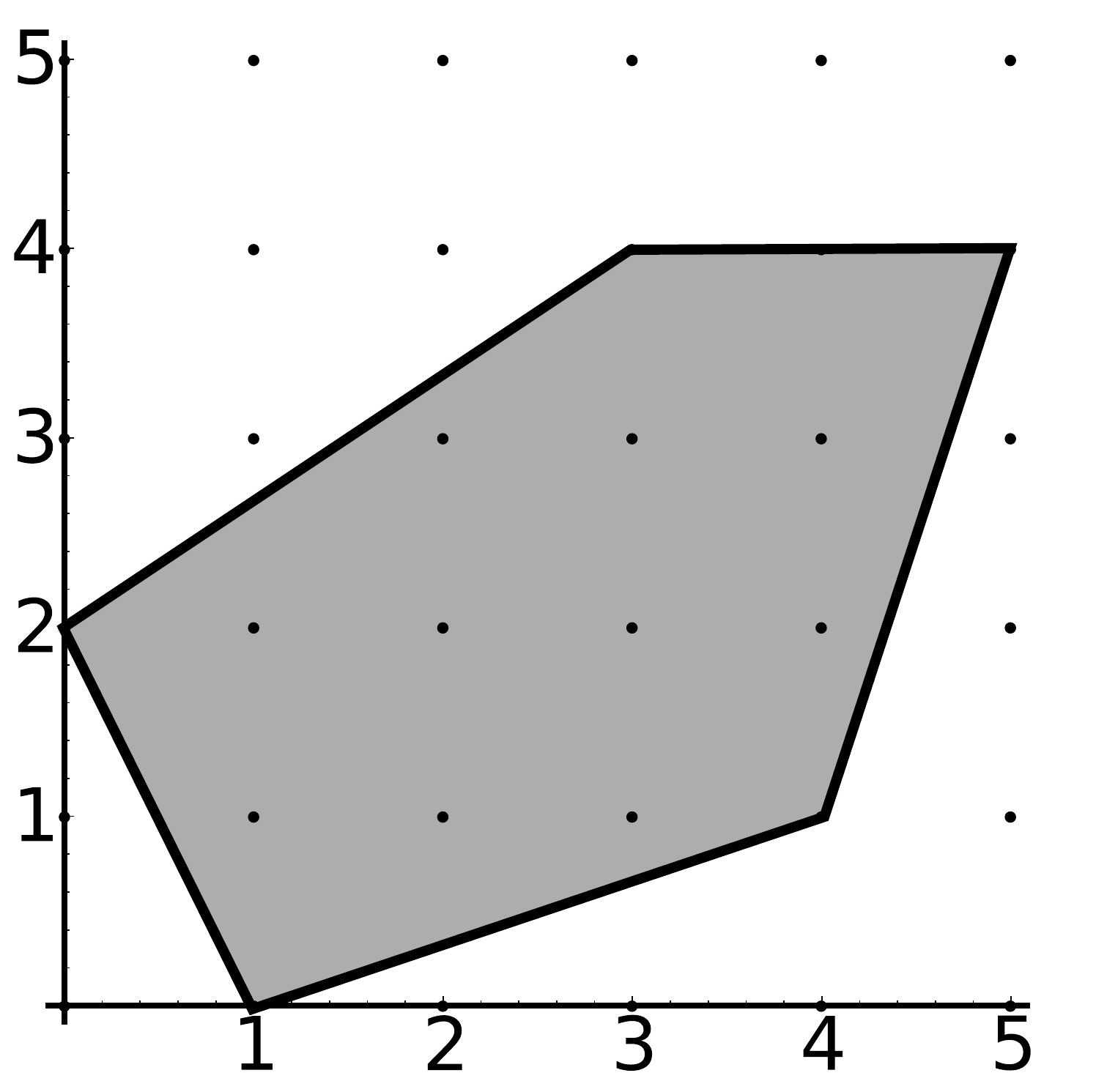}
	  
      \label{fig:boundedsystem}
	}
	\subfloat[]{
	  \includegraphics[width=0.32\textwidth]{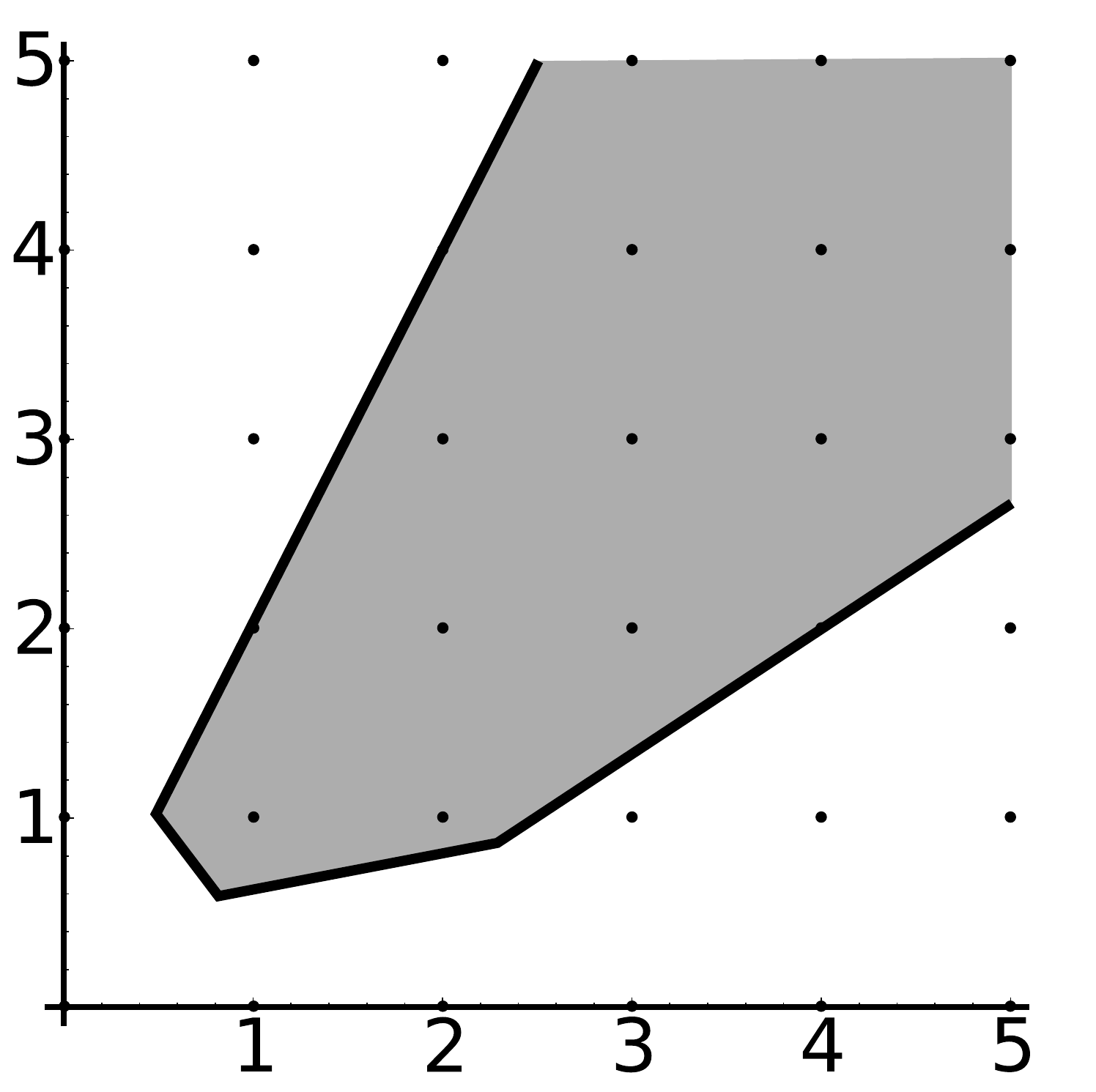}
      \label{fig:absunboundedsystem}
	}
	\caption{\textbf{a:} a partially bounded system; the directions $h = (-1,1)^T$ and $-h$ are the only bounded directions in the example. \textbf{b:} a bounded system; all directions are bounded. \textbf{c:} an absolutely unbounded system; all directions are unbounded.}
\end{figure}

Branch-and-bound alone is already complete on bounded constraint systems, i.e., systems where all directions are bounded:

\begin{definition}[Bounded Direction]
A \emph{direction}/vector $h \in \mathbb{Q}^n \setminus \{0^n\}$ is \emph{bounded} in the constraint system $A x \leq b$ if there exist $l, u \in \mathbb{Q}$ such that $A x \leq b$ implies $h^T x \leq u$ and $-h^T x \leq -l$. Otherwise, it is called unbounded. (See Figure~\ref{fig:boundeddirection} for an example.)
\label{def:boundeddir}
\end{definition}

\begin{definition}[Bounded System]
A constraint \emph{system} $A x \leq b$ is \emph{bounded} if all directions $h \in \mathbb{Q}^n \setminus \{0^n\}$ are bounded. Otherwise, it is called unbounded. (See Figure~\ref{fig:boundedsystem} for an example.)
\label{def:boundedsys}
\end{definition}

For bounded systems, branch-and-bound is one of the most popular and efficient algorithms. 
It may, however, diverge if the system has unbounded directions. 
Even so, not all unbounded systems are equally difficult. 
For instance, a system where all directions are unbounded has always a mixed solution:

\begin{lemma}[Absolutely Unbounded~\cite{BrombergerWeidenbach:16}]
If all directions are unbounded in a constraint \emph{system} $A x \leq b$, then the constraint system has an integer solution. (See Figure~\ref{fig:absunboundedsystem} for an example.)
\label{lemma:absoluteunbounded}
\end{lemma}

In a previous article, we described two cube tests that detect and solve constraint systems with infinite lattice width (another name for absolutely unbounded systems) in polynomial time~\cite{BrombergerWeidenbach:16}. 
The case of absolutely unbounded systems is, therefore, trivial and branch-and-bound can be easily extended so it also becomes complete for absolutely unbounded systems. 
The actual difficult case is when some directions are bounded and others unbounded. 
We call these systems \emph{partially unbounded}. 
Here, branch-and-bound and most other algorithms diverge or become inefficient in practice. 
The transformations, which we present, are designed to efficiently handle this subclass of problems.

\section{Mixed-Echelon-Hermite Transformation}
\label{SE:MixedEchelonHermiteTransformation}
\label{SE:Mixed-Echelon-Hermite Transformation}
\label{SE:doublebounded}
\label{SE:Double-Bounded Constraint Systems}

Our overall goal is to present an equisatisfiable transformation that turns any constraint system into a system that is bounded, i.e., a system on which branch-and-bound and many other arithmetic decision procedures terminate. 
In this section, we only present such a transformation for a subset of constraint systems, which we call \emph{double-bounded constraint systems}. 
We then show in the next section that each constraint system can be reduced to an equisatisfiable double-bounded system. 
We also show how to efficiently transform a mixed solution from the double-bounded reduction to a mixed solution for the original system.

\begin{definition}[Double-Bounded Constraint System]
A constraint system $D x \leq u$ is \emph{double-bounded} if $D x \leq u$ implies $D x \geq l$ for $l \in \mathbb{Q}^m$. 
For such a double-bounded system, we call the bounds $u$ the \emph{upper bounds} of $D x$ and the bounds $l$ the lower bounds of $D x$. 
Moreover, we typically write $l \leq D x \leq u$ instead of $D x \leq u$ although the \emph{lower bounds} $l$ are only implicit.
\end{definition}

Note that only the inequalities in a double-bounded constraint system are guaranteed to be bounded. Variables might still be unbounded. 
For instance, in the constraint system $1 \leq 3 x_1 - 3 x_2 \leq 2$ both inequalities are bounded but the variables $x_1$ and $x_2$ are not. 
Moreover, the above constraint system is also an example where branch-and-bound diverges. 
This means that even bounding all inequalities does not yet guarantee termination. 
So for our purposes, a double-bounded constraint system is still too complex. 

This changes, however, if we also require that the coefficient matrix $D$ of our constraint system is a \emph{lower triangular matrix with gaps}:

\begin{definition}[Lower Triangular Matrix with Gaps]
A matrix $A \in \mathbb{Q}^{m \times n}$ is \emph{lower triangular with gaps} if it holds for each column $j$ that $\piv{A}{j} > m$ or that $\piv{A}{j} < \piv{A}{k}$ for all columns $k$ with $j < k \leq n$, i.e., column $j$ either has only zero entries or all pivoting entries right of $j$ have a higher row index.\label{def:lowertriangulargap}
\end{definition}

A matrix is lower triangular if and only if the row indices of its pivots are strictly increasing, i.e., $\piv{A}{1} < \ldots < \piv{A}{n}$. 
If we also allow it to have gaps, only the row indices of pivots with non-zero columns have to be strictly increasing. 
Now we get termination for free because of our restrictions:

\begin{lemma}[Lower Triangular Double-Bounded Systems]
Let $D \in \mathbb{Q}^{m \times n}$ be a lower triangular matrix with gaps and $l \leq D x \leq u$ be a double-bounded constraint system. Then each variable $x_j$ is either bounded, i.e., $l \leq D x \leq u$ implies that $l'_j \leq x_j \leq u'_j$ or its column in $D$ has only zero entries.\label{lemma:LTDBS}
\end{lemma}
\begin{proof}
Proof by induction. 
Assume that the above property already holds for all variables $x_k$ with $k < j$. 
Let $p = \piv{D}{j}$. If $p > m$, then the column $j$ of $D$ is zero and we are done.
If $p \leq m$, then the pivoting entry $d_{pj}$ of column $j$ is non-zero. 
Because of Definition~\ref{def:lowertriangulargap} and our induction hypothesis, this also means that each column $k$ with $k < j$ has either a zero entry in row $p$ or the variable $x_k$ is bounded by our induction hypothesis, i.e., $l \leq D x \leq u$ implies $l'_k \leq x_k \leq u'_k$. 
Since Definition~\ref{def:lowertriangulargap} also implies that row $p$ has only zero entries to the right of $d_{pj}$, 
the row $p$ has only one unbounded variable with a non-zero entry, viz., $x_j$. 
This means we can transform the row $l_p \leq d_{p}^T x \leq u_p$ into the following two inequalities: 
$l_p - \sum_{k = 1}^{j-1}d_{pk} x_k \leq d_{pj} x_j$ and $u_p - \sum_{k = 1}^{j-1}d_{pk} x_k \geq d_{pj} x_j$, where the variables $x_k$ on the left sides are either bounded or $d_{pk} = 0$. 
Hence, we can derive an upper and lower bound for $x_j$ via bound propagation/refinement~\cite{JovanovicdeMoura:13}.
\end{proof}

\begin{corollary}[BnB-LTDB-Termination]
Branch-and-bound terminates on every double-bounded system $l \leq D x \leq u$ where $D$ is lower triangular with gaps.\label{corollary:BNBLTDB}
\end{corollary}

Our next goal is to efficiently transform every double-bounded system $l \leq D x \leq u$ into an equisatisfiable system that also has a lower triangular coefficient matrix with gaps. 
We start by defining a class of transformations that do not only preserve mixed equisatisfiability, but are also very expressive.

\begin{definition}[Mixed Column Transformation Matrix~\cite{ChristHoenicke2015}]
Given a mixed constraint system. 
A matrix $V \in \mathbb{Q}^{n \times n}$ is a \emph{mixed column transformation matrix} if it is invertible and consists of an invertible matrix $V_{(\mathbb{Q})} \in \mathbb{Q}^{n_1 \times n_1}$, a unimodular matrix $V_{(\mathbb{Z})} \in \mathbb{Z}^{n_2 \times n_2}$, and a matrix $V_{(M)} \in \mathbb{Q}^{n_1 \times n_2}$ such that
\centermath{V = \left( \begin{array}{l l}
V_{(\mathbb{Q})} &V_{(M)}\\
0^{n_2 \times n_1} &V_{(\mathbb{Z})}
\end{array} \right).}\label{def:MCTM}
\end{definition}

The formal definition of mixed column transformation matrices may seem anything but intuitive. 
However, they actually describe a straightforward class of transformations, viz., any combination of \emph{mixed equisatisfiable column transformations} that can be performed on a matrix $A$. 
Mixed equisatisfiable column transformations are either (i) multiplying a column by $-1$; (ii) the swapping of two columns $i, j$ of the same type, i.e., $i,j \leq n_1$ or $i,j > n_1$; (iii) multiplying a rational column $j$ (i.e., $j \leq n_1$) with a non-zero rational factor; (iv) adding a rational multiple of a column $j$ with $j \leq n_1$ to any other column $i$; and (v) adding an integer multiple of a column $j > n_1$ to a different column $i$ with $i > n_1$.
If we perform the same mixed equisatisfiable column transformations that resulted in $H$ from $A$ to an $n \times n$ identity matrix, then the transformed identity matrix $V$ is a mixed column transformation matrix and $H = A V$. 
We just compacted the column transformations into a matrix. 
We cannot just use $V$ to redo the column transformations, but we can also use its inverse to undo them:

\begin{lemma}[Mixed Column Transformation Inversion~\cite{ChristHoenicke2015}]
Given a mixed constraint system. 
Let $V \in \mathbb{Q}^{n \times n}$ be a mixed column transformation matrix. 
Then $V^{-1}$ is also a mixed column transformation matrix.\label{lemma:MCTINV}
\end{lemma}

This means that each mixed column transformation matrix defines a bijection from $(\mathbb{Q}^{n_1} \times \mathbb{Z}^{n_2})$ to $(\mathbb{Q}^{n_1} \times \mathbb{Z}^{n_2})$. 
Hence, they guarantee mixed equisatisfiability:

\begin{lemma}[Mixed Column Transformation Equisatisfiability~\cite{ChristHoenicke2015}]
Let $A x \leq b$ be a mixed constraint system. Let $V \in \mathbb{Q}^{n \times n}$ be a mixed column transformation matrix. Then every solution $y \in \Mset{(A V) y \leq b})$ can be converted into a solution $V y = x \in \Mset{A x \leq b}$ and vice versa.\label{lemma:MCTEQS}
\end{lemma}

Moreover, the mixed column transformation matrix $V$ also establishes a direct relationship between the linear combinations of the original constraint system and the transformed one:

\begin{lemma}[Mixed Column Transformation Implications]
Let $A x \leq b$ be a constraint system.
Let $V \in \mathbb{Q}^{n \times n}$ be a mixed column transformation matrix. 
Let $A x \leq b$ imply $h^T x \leq g$. Then $A V z \leq b$ implies $h^T V z \leq g$.\label{lemma:MCTIMP}
\end{lemma}
\begin{proof}
By Lemma~\ref{lemma:farkasimplies}, $A x \leq b$ implies $h^T x \leq g$ iff there exists a non-negative linear combination $y \in \mathbb{Q}^n$ such that $y \geq 0$, $y^T A = h^T$ and $y^T b \leq g$. 
Multiplying $y^T A = h^T$ with $V$ results in $y^T A V = h^T V$ and thus $y$ is also the non-negative linear combination of inequalities $A V z \leq b$ that results in $h^T V z \leq g$.
\end{proof}

\begin{corollary}[Mixed Column Transformation Certificates]
Let $A x \leq b$ be a constraint system.
Let $V \in \mathbb{Q}^{n \times n}$ be a mixed column transformation matrix. 
Then $y$ is a certificate of unsatisfiability for $A x \leq b$ iff it is one for $A V z \leq b$.\label{corollary:MCTCERT}
\end{corollary}

Now we only need a mixed column transformation matrix $V$ for every coefficient matrix $A$ such that $H = AV$ is lower triangular with gaps. 
One such matrix $V$ is the one that transforms $A$ into \emph{Mixed-Echelon-Hermite normal form}:

\begin{definition}[Mixed-Echelon-Hermite Normal Form~\cite{ChristHoenicke2015}]
A matrix $H \in \mathbb{Q}^{m \times n}$ is in \emph{Mixed-Echelon-Hermite normal form} if 
\centermath{H = \left( \begin{array}{c c l c c}
E & \; & 0^{r \times (n_1 - r)} & \; & 0^{r \times n_2} \\
E' & \; & 0^{(m - r) \times (n_1 - r)} & \; & H' \\
\end{array} \right),}
where $E$ is an $r \times r$ identity matrix (with $r \leq n_1$), $E' \in \mathbb{Q}^{(m-r) \times r}$, and $H' \in \mathbb{Q}^{(m-r) \times n_2}$ is a matrix in hermite normal form, i.e., a lower triangular matrix without gaps, where each entry $h'_{{\piv{H'}{j}}k}$ in the row $\piv{H'}{j}$ is non-negative and smaller than $h'_{{\piv{H'}{j}}j}$).\label{def:MEHNF}
\end{definition}

The following proof for the existence of the Mixed-Echelon-Hermite normal form is constructive and presents the Mixed-Echelon-Hermite transformation.

\begin{lemma}[Mixed-Echelon-Hermite Transformation]
Let $A \in \mathbb{Q}^{m \times n}$ be a matrix, where the upper left $r \times n_1$ submatrix has the same rank $r$ as the complete left $m \times n_1$ submatrix. Then there exists a mixed transformation matrix $V \in \mathbb{Q}^{n \times n}$ such that $H = AV$ is in \emph{Mixed-Echelon-Hermite normal form}.\label{lemma:MEHNF}
\end{lemma}
\begin{proof}
Proof from~\cite{ChristHoenicke2015} with slight modifications so it also works for singular matrices. 
We subdivide $A$ into
\centermath{A = \left( \begin{array}{c c}
A_{11} & A_{12} \\
A_{21} & A_{22} \\
\end{array} \right) }
such that $A_{11} \in \mathbb{Q}^{r \times n_1}$, $A_{12} \in \mathbb{Q}^{r \times n_2}$, $A_{21} \in \mathbb{Q}^{m-r \times n_1}$, and $A_{21} \in \mathbb{Q}^{m-r \times n_2}$. 
Then we bring $A_{11}$ with an invertible matrix $V_{11} \in \mathbb{Q}^{n_1 \times n_1}$ into reduced echelon column form $H_{11} = (E \; 0^{r \times (n_1 - r)}) = A_{11} V_{11}$, where $E$ is an $r \times r$ identity matrix. 
We get $V_{11}$ and $H_{11}$ by using Bareiss algorithm instead of the better known Gaussian elimination as it is polynomial in time~\cite{Bareiss1968}.\footnote{\label{footnote:transformimpl}In our implementation, we do actually use less efficient, Gaussian elimination based transformations. 
The reason is that these transformations are incrementally efficient (see Appendix~\ref{SE:incremental}). 
Our experiments show that the transformation cost still remains negligible in practice.}
Note that the last $n_1-r$ columns of $H_{21} = (H'_{21}\; 0^{(m - r) \times (n_1 - r)}) = A_{21} V_{11}$ are also zero because all rows in $A_{21}$ are linear dependent of $A_{11}$ (due to the rank). 
Next we notice that 
\centermath{A_{12} - A_{11} V_{11} \left( \begin{array}{c} A_{12} \\ 0^{(n_1-r) \times n_2} \end{array} \right) = A_{12} - (E \; 0^{r \times (n_1 - r)}) \left( \begin{array}{c} A_{12} \\ 0^{(n_1-r) \times n_2} \end{array} \right) = 0^{r \times n_2}}
so we can reduce the upper right submatrix $A_{12}$ to zero by adding multiples of the $n_1$ columns with rational variables to the $n_2$ columns with integer variables. 
However, this also transforms the lower right submatrix $A_{22}$ into 
\centermath{H'_{22} = A_{22} - A_{21} V_{11} \left( \begin{array}{c} A_{12} \\ 0^{(n_1-r) \times n_2} \end{array} \right). }
Finally, we transform this new submatrix $H'_{22}$ into hermite normal form $H_{22}$ via the algorithm of Kannan and Bachem (or a similar polynomial time algorithm).\footnotemark[5] 
This algorithm also returns a unimodular matrix $V_{22} \in \mathbb{Z}^{n_2 \times n_2}$ such that $H_{22} = H'_{22} V_{22}$. 
To summarize: our total mixed transformation matrix is\\
$V = \left( \begin{array}{l c}
V_{11} & -V_{11} \cdot \left( \begin{array}{c} A_{12} \\ 0^{(n_1-r) \times n_2} \end{array} \right) \cdot V_{22} \\
0^{n_2 \times n_1} & V_{22} \\
\end{array} \right)$  and $H = AV = \left( \begin{array}{c c}
H_{11} & 0^{r \times n_2} \\
H_{21} & H_{22} \\
\end{array} \right).$
\end{proof}

It is not possible to transform every matrix $A \in \mathbb{Q}^{m \times n}$ into Mixed-Echelon-Hermite normal form. 
We have to restrict ourselves to matrices, where the upper left $r \times n_1$ submatrix has the same rank $r$ as the complete left $m \times n_1$ submatrix. 
However, this is very easy to accomplish for a system of linear mixed arithmetic constraints $l \leq A x \leq u$. 
The reason is that the order of inequalities does not change the set of satisfiable solutions. 
Hence, we can swap the inequalities and, thereby, the rows of $A$ until its upper left $r \times n_1$ submatrix has the desired form. 
This also means that there are usually multiple possible inequality orderings that each have their own Mixed-Echelon-Hermite normal form $H$.

To conclude this section: whenever we have a double-bounded constraint system $l \leq D x \leq u$, we can transform it (after some row swapping) into an equisatisfiable system $l \leq H y \leq u$ where $H = D V$ is in Mixed-Echelon-Hermite normal form and $V y = x$. 
Since $H$ is also a lower triangular matrix with gaps, branch-and-bound terminates on $l \leq H y \leq u$ with a mixed solution $t$ or it will return unsatisfiable (Corollary~\ref{corollary:BNBLTDB}). 
Moreover, we can convert any mixed solution $t$ for $l \leq H y \leq u$ into a mixed solution $s$ for $l \leq D x \leq u$ by setting $s := V t$.
Hence, we have a complete algorithm for double-bounded constraint systems.

\section{Double-Bounded Reduction}
\label{SE:Double-Bounded Reduction}
\label{SE:DoubleBoundedReduction}
\label{SE:Splitting Bounded and Unbounded Inequalities}
\label{SE:splitting}

In the previous Section, we have shown how to solve a double-bounded constraint system. 
Now we show how to reduce any constraint system $A' x \leq b'$ to an equisatisfiable double-bounded system $l \leq D x \leq u$. 
Moreover, we explain how to take any solution of $l \leq D x \leq u$ and turn it into a solution for $A' x \leq b'$. 

\begin{figure}[t]
    \subfloat[]{
	  \includegraphics[width=0.32\textwidth]{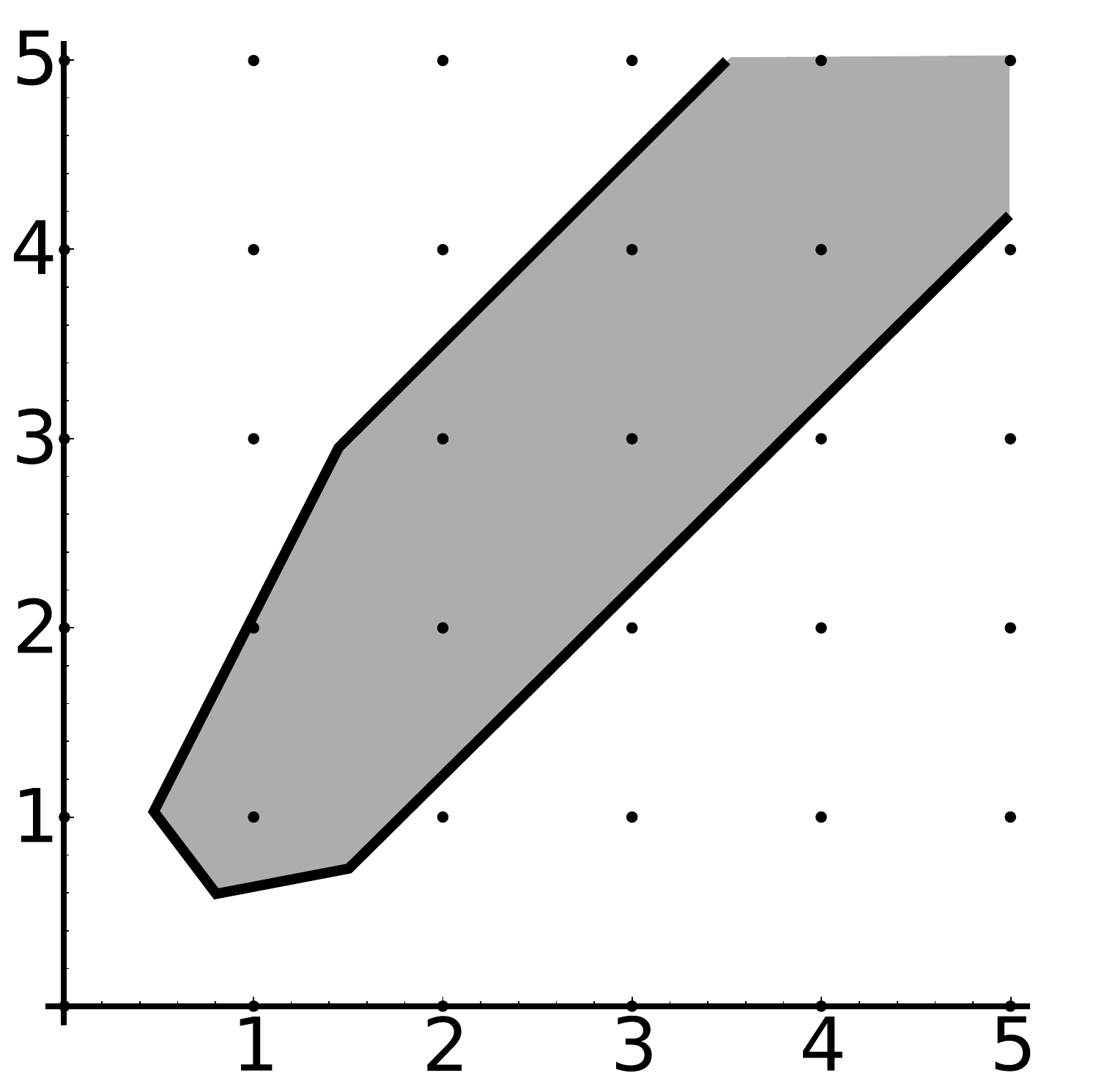}
      \label{fig:splitsystem}
	}
	\subfloat[]{
      \includegraphics[width=0.32\textwidth]{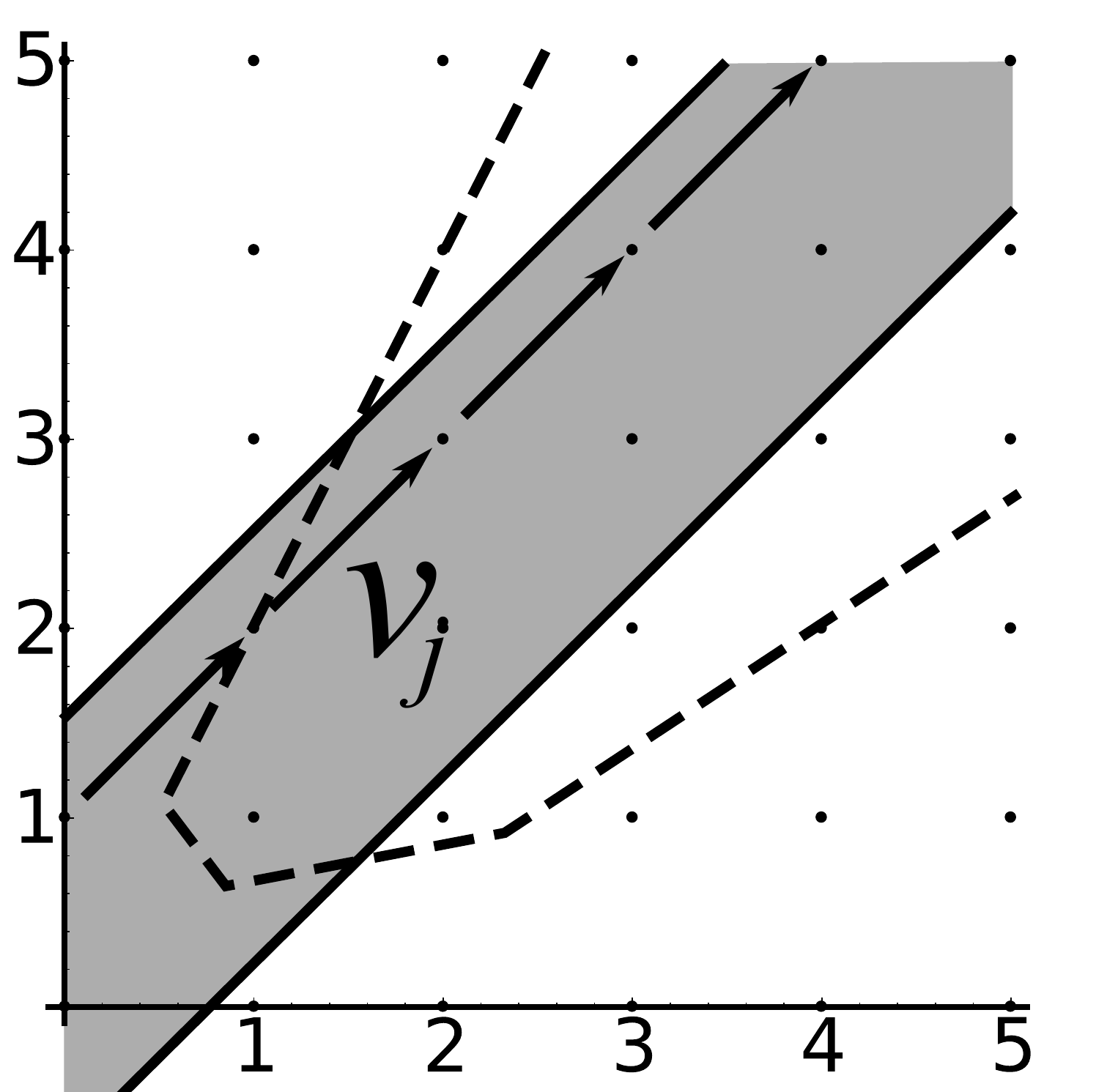}
	  
      \label{fig:boundedpart}
	}
	\subfloat[]{
	  \includegraphics[width=0.32\textwidth]{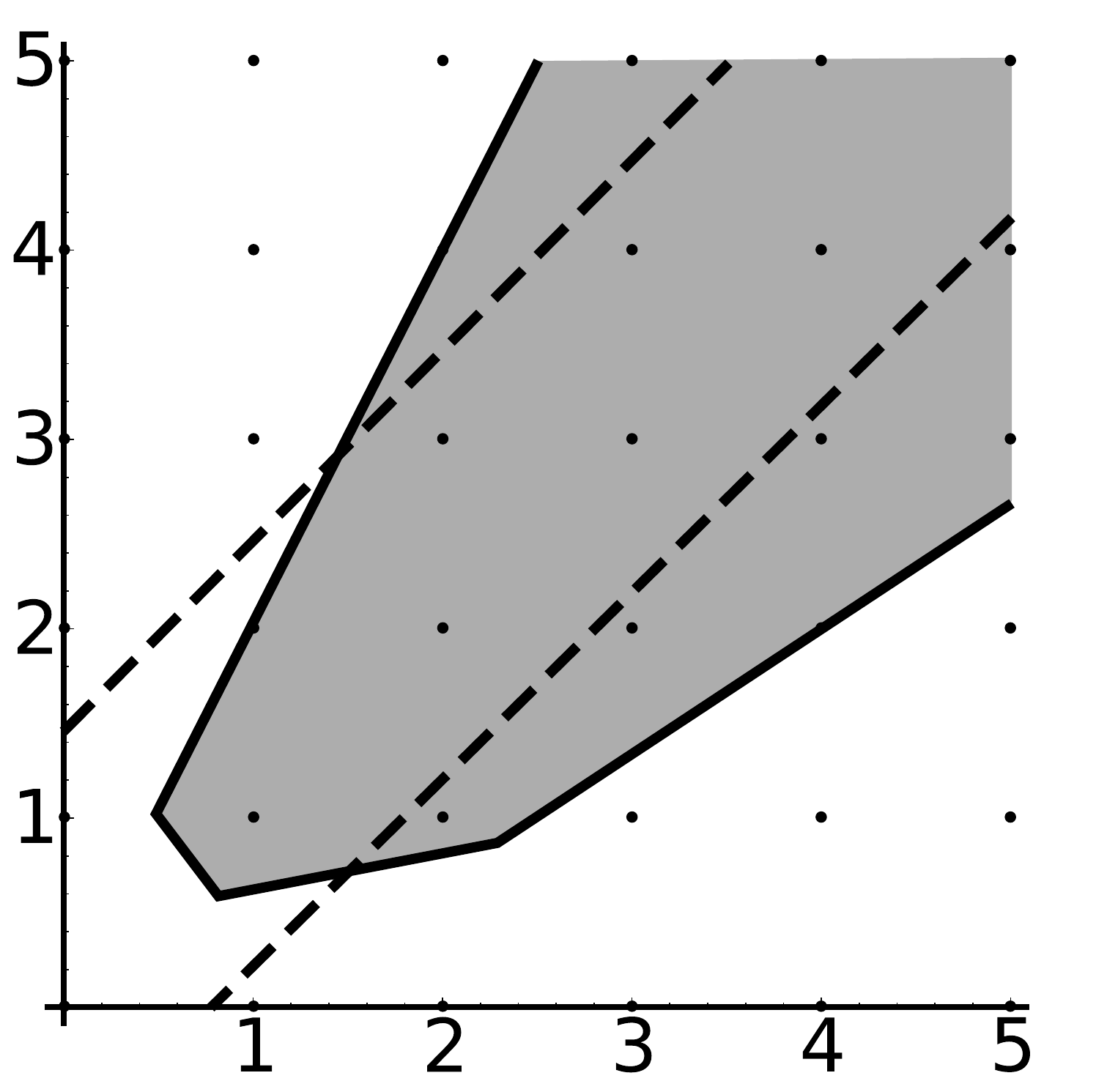}
      \label{fig:unboundedpart}
	}
	\caption{\textbf{a:} a partially bounded system. \textbf{b:} the (double-)bounded part of (a); $v_j := (1,1)^T$ is an orthogonal direction to the bounding row vectors. \textbf{c:} the unbounded part of (a).}
\end{figure}

As the first step of our reduction, we reformulate the constraint system into a so called \emph{split system}:

\begin{definition}[Split System]
$(A x \leq b) \cup (l \leq D x \leq u)$ is a \emph{split system} if: (i) all directions are unbounded in $A x \leq b$; (ii) all row vectors $a_i$ from $A$ are also unbounded in $(A x \leq b) \cup (l \leq D x \leq u)$. Moreover, we call $A x \leq b$ the unbounded part and $l \leq D x \leq u$ the bounded part of the split system.
\end{definition}

A split system consists of an unbounded part $A x \leq b$ that is guaranteed to have (infinitely many) integer solutions (see Lemma~\ref{lemma:absoluteunbounded}) and a double-bounded part $l \leq D x \leq u$. 
Any constraint system can be brought into the above form (see Figures~\ref{fig:splitsystem}---\ref{fig:unboundedpart} for an example). 
We just have to move all unbounded inequalities into the unbounded part and all bounded inequalities into the bounded part. 

\begin{lemma}[Split Equivalence]
Let $A' x \leq b'$ be a constraint system with $A' \in \mathbb{Q}^{m \times n}$. Then there exists an equivalent split system $(A x \leq b) \cup (l \leq D x \leq u)$ where: (i) $A \in \mathbb{Q}^{m_1 \times n}$ and $D \in \mathbb{Q}^{m_2 \times n}$ so that $m_1 + m_2 = m$; (ii) all rows $d_i^T$ of $D$ and $a^T_k$ of $A$ appear as rows in $A'$; and (iii) $D x \leq u$ implies $l \leq D x$.\label{lemma:splitequivalence}
\end{lemma}
\begin{proof}
For (i), (ii), and the equivalence, it is enough to move all bounded inequalities $a'^T_i x \leq b'_i$ of $A' x \leq b'$ into $D x \leq u$ and all unbounded inequalities into $A x \leq b$. 
For (iii), we assume for a contradiction that $D x \leq u$ does not imply $l_i \leq d_i^T x$ but $(D x \leq u) \cup (A x \leq b)$ does. 
By Lemma~\ref{lemma:farkasimplies}, this means that there exists a $y \in \mathbb{Q}^{m_2}$ with $y \geq 0^{m_2}$ and a $z \in \mathbb{Q}^{m_1}$  with $z \geq 0^{m_1}$ so that $y^T D + z^T A = -d_i^T$ and $y^T u + z^T b \leq -l_i$. 
We also know that there exists a $z_k > 0$ because $D x \leq u$ alone does not imply $l_i \leq d_i^T x$. 
We use this fact to reformulate $y^T D + z^T A = -d_i^T$ into  $- a^T_k = \frac{1}{z_k}\left[y^T D + d_i^T + \sum_{j = 1, j \neq k}^{m_1} z_j a^T_j  \right],$ 
and use the bounds of the inequalities in $D x \leq u$ and $A x \leq b$ to derive a lower bound for  $a^T_k x$:
$ - a^T_k x \leq \frac{1}{z_k}\left[y^T u + u_i + \sum_{j = 1, j \neq k}^{m_1} z_j b_j \right].$ 
Hence, $a^T_k$ is bounded in $A' x \leq b'$ and we have our contradiction.
\end{proof}

The above Lemma also shows that the bounded part of a constraint system is self-contained, i.e., a constraint system implies that a direction is bounded if and only if its bounded part does. 
The actual difficulty of reformulating a system into a split system is not the transformation per se, but finding out which inequalities are bounded or not. 
There are many ways to detect whether an inequality is bounded by a constraint system. 
Most work even in polynomial time. 
For instance, solving the linear rational optimization problem ``minimize $a_i^T x$ such that $A x \leq b$'' returns $-\infty$ if $a_i$ is unbounded, $\infty$ if $A x \leq b$ has no rational solution, and the optimal lower bound $l_i$ for $a_i^T x$ otherwise. 
However, it still requires us to solve $m$ linear optimization problems. 

A, in our opinion, more efficient alternative is based on our previously presented algorithm for finding equality bases~\cite{BrombergerWeidenbach:16a}. 
This is due to the following relationship between bounded directions and equalities:

\begin{lemma}[Bounds and Equalities]
Let $\Pset{A}{b} \neq \emptyset$. Then $h$ is bounded in $A x \leq b$ iff $A x \leq 0^m$ implies that $h^T x = 0$.\label{lemma:boundsandequalities}
\end{lemma}
\begin{proof}
By Definition~\ref{def:boundeddir}, $h$ is bounded in $A x \leq b$ means that there exists $l, u \in \mathbb{Q}$ such that $A x \leq b$ implies $h^T x \leq u$ and $-h^T x \leq -l$.
By Lemma~\ref{lemma:farkasimplies}, this is equivalent to: there exist $l, u \in \mathbb{Q}$, $y, z \in \mathbb{Q}^m$ with $y, z \geq 0^m$, and $y^T A = h^T = -z^T A$ so that $y^T b \leq u$ and $z^T b \leq -l$. 
Symmetrically, $A x \leq 0$ implies that $h^T x = 0$ is equivalent to: there exist a $y, z \in \mathbb{Q}^m$ with $y, z \geq 0^m$ and $y^T A = h^T = -z^T A$ so that $y^T 0^m \leq 0$ and $z^T 0^m \leq 0$. 
Since $u$ and $l$ only have to exists, we can trivially choose them as $u := y^T b$ and $l := -z^T b$. 
This means that $y^T b \leq u$, $z^T b \leq -l$, $y^T 0^m \leq 0$, and $z^T 0^m \leq 0$ are all trivially satisfied by any pair of linear combinations $y, z \in \mathbb{Q}^m$ with $y, z \geq 0^m$ such that $y^T A = h^T = -z^T A$. 
Hence, the two definitions are equivalent and our lemma holds.
\end{proof}

It is easy and efficient to compute an equality basis for $A x \leq 0^m$ and to determine with it the inequalities in $A x \leq b$ that are bounded~\cite{BrombergerWeidenbach:16a}. 
The only disadvantage towards the optimization approach is that we do not derive an optimal lower bound $l$ for the inequalities. 
This is no problem because only the existence of lower bounds is relevant and not the actual bound values.

In a split system $(A x \leq b) \cup (l \leq D x \leq u)$, 
the unbounded part is actually inconsequential to the rational/mixed satisfiability of the system. 
It may reduce the number of rational/mixed solutions, but it never removes them all. 
Hence, $(A x \leq b) \cup (l \leq D x \leq u)$ is equisatisfiable to just $l \leq D x \leq u$. 
We first show this equisatisfiability for the rational case:

\begin{lemma}[Rational Extension]
Let $(A x \leq b) \cup (l \leq D x \leq u)$ be a split system. 
Let $s \in \mathbb{Q}^n$  be a rational solution to the bounded part $l \leq D x \leq u$ such that $D s = g$, where $g \in \mathbb{Q}^{m_2}$. 
Then $(A x \leq b) \cup (D x = g)$ has a solution $s'$.\label{lemma:rationalsoundness}
\end{lemma}
\begin{proof}
Assume for a contradiction that $(A x \leq b) \cup (D x = g)$ has no solution. 
By Lemma~\ref{lemma:farkasunsat}, this means that there exist a $y \in \mathbb{Q}^{m_1}$ with $y \geq 0^{m_1}$ and $z, z' \in \mathbb{Q}^{m_2}$ with $z, z' \geq 0^{m_2}$ such that $y^T A + z^T D - z'^T D = 0^n$ and $y^T b + z^T g - z'^T g < 0$. 
Since $D x = g$ is satisfiable by itself, there must exist a $y_i > 0$. 
Now we use this fact to reformulate the equation $y^T A + z^T D - z'^T D = 0^n$ into 
\centermath{-a_i^T = \frac{1}{y_i}\left[\left(\sum_{j = 1 \\ j \neq i}^{m_1} y_j a_j^T \right) + z^T D - z'^T D \right],} 
from which we deduce a lower bound for $a_i^T x$ in $(A x \leq b) \cup (l \leq D x \leq u)$:
\centermath{-a_i^T x \leq \frac{1}{y_i}\left[\left(\sum_{j = 1 \\ j \neq i}^{m_1} y_j b_j \right) + z^T u - z'^T l \right].}
Therefore, $a_i$ is bounded in $(A x \leq b) \cup (l \leq D x \leq u)$, which is a contradiction.
\end{proof}

Note that the bounded part $l \leq D x \leq u$ of a split system can still have unbounded directions (not inequalities). 
Some of these unbounded directions in $l \leq D x \leq u$ are the orthogonal directions to the row vectors $d_i$, i.e., vectors $v_j \in \mathbb{Z}^n$ such that $d_i^T v_j = 0$ for all $i \in \{1,\ldots,m_2\}$. 
This also means that the existence of one mixed solution $s \in (\mathbb{Q}^{n_1} \times \mathbb{Z}^{n_2})$ and one unbounded direction proves the existence of infinitely many mixed solutions. 
We just need to follow the orthogonal directions, i.e., for all $\lambda \in \mathbb{Z}$, $s' = \lambda \cdot v_j + s$ is also a mixed solution because $d_i^T s' = \lambda \cdot d_i^T v_j + d_i^T s = d_i^T s$. 
(See Figures~\ref{fig:splitsystem}---\ref{fig:unboundedpart} for an example.)
In the next two steps, we prove that $A x \leq b$ cannot cut off all of these orthogonal solutions because it is completely unbounded. 
The first step proves that $A x \leq b$ remains absolutely unbounded even if we settle on one set of orthogonal solutions, i.e., enforce $D x = D s$ for some solution $s$.

\begin{lemma}[Persistence of Unboundedness]
Let $(A x \leq b) \cup (l \leq D x \leq u)$ be a split system. 
Let $s \in \mathbb{Q}^{n}$ be a rational solution for $l \leq D x \leq u$ such that $D s = g$ (with $g \in \mathbb{Q}^{m_2}$). 
Then all row vectors $a_i$ from $A$ are still unbounded in $(A x \leq b) \cup (D x = g)$.\label{lemma:persistenceunboundedness}
\end{lemma}
\begin{proof}
By Lemma~\ref{lemma:rationalsoundness}, $(A x \leq b) \cup (D x = g)$ has at least a rational solution $s^*$.
Moreover, $(A x \leq 0) \cup (D x = 0)$ does not imply $a_i^T x = 0$ because of Lemma~\ref{lemma:boundsandequalities} and the assumption that the row vectors $a_i$ from $A$ are unbounded in $(A x \leq b) \cup (l \leq D x \leq u)$.
In reverse, $(A x \leq b) \cup (D x = g)$ having a real solution, $(A x \leq 0) \cup (D x = 0)$ does not imply $a_i^T x = 0$, and Lemma~\ref{lemma:boundsandequalities} prove together that the row vectors $a_i$ from $A$ are also unbounded in $(A x \leq b) \cup (D x = g)$.
\end{proof}

The next step proves how to extend the mixed solution from the bounded part to the complete system with the help of the Mixed-Echelon-Hermite normal form and the absolute unboundedness of $A x \leq b$.

\begin{lemma}[Mixed Extension]
Let $(A x \leq b) \cup (l \leq D x \leq u)$ be a split system. 
Let $s \in (\mathbb{Q}^{n_1} \times \mathbb{Z}^{n_2})$ be a mixed solution for $l \leq D x \leq u$. 
Then  $(A x \leq b) \cup (l \leq D x \leq u)$ has a mixed solution $s'$.\label{lemma:mixedsoundness}
\end{lemma}
\begin{proof}
Let $g = D s$. 
Without loss of generality we assume that the upper left $r \times n_1$ submatrix of $D$ has the same rank $r$ as the complete left $m_1 \times n_1$ submatrix of $D$. 
(Otherwise, we just reorder the rows accordingly.) 
Therefore, there exists a mixed column transformation matrix $V$ such that $H = D V$ is in mixed-echelon-hermite normal form (see Lemma~\ref{lemma:MEHNF}). 
By Lemma~\ref{lemma:MCTEQS}, there exists a mixed vector $t \in (\mathbb{Q}^{n_1} \times \mathbb{Z}^{n_2})$ such that $s = V t$ and $t$ is a mixed-solution to $l \leq H y \leq u$ as well as $H y = g$.
Let $\mathcal{U}$ be the set of indices with $0$ columns in $H$ and $\mathcal{B}$ the column indices with bounded variables. 
Then the equation system $(H y = g)$ fixes each variable $y_j$ with $j \in \mathcal{B}$ to the value $t_j$ because $H$ is lower triangular with gaps. 
Hence, $((A V) y \leq b) \cup (H y = g)$ is equivalent to 
\centerequation{ A \left[ \sum_{j \in \mathcal{U}} \left( \begin{array}{c}
v_{1j}\\
\vdots\\
v_{nj}
\end{array} \right) \cdot y_j \right] \leq b - A \left[ \sum_{j \in \mathcal{B}} \left( \begin{array}{c}
v_{1j}\\
\vdots\\
v_{nj}
\end{array} \right) \cdot t_j \right]. \label{eq:fixedboundedinunbounded}}
Due to Lemma~\ref{lemma:persistenceunboundedness} and~\ref{lemma:MCTIMP}, all directions are unbounded in (\ref{eq:fixedboundedinunbounded}). 
This means (\ref{eq:fixedboundedinunbounded}) has an integer solution (Lemma~\ref{lemma:absoluteunbounded}) assigning each variable $y_j$ with $j \in \mathcal{U}$ to a $t'_j \in \mathbb{Z}$. (Can be computed via the unit cube test~\cite{BrombergerWeidenbach:17}).
We extend this solution to all variables $y$ by setting $t'_j := t_j$ for $j \in \mathcal{B}$ and we have a mixed solution $t' \in (\mathbb{Q}^{n_1} \times \mathbb{Z}^{n_1})$ for $((A V) y \leq b) \cup (l \leq H y \leq u)$.
Hence, we have via Lemma~\ref{lemma:MCTEQS} a mixed solution $s' \in (\mathbb{Q}^{n_1} \times \mathbb{Z}^{n_2})$ for $(A x \leq b) \cup (l \leq D x \leq u)$ with $s' = V t'$.
\end{proof}


\begin{corollary}[Double-Bounded Reduction]
The split system $(A x \leq b) \cup (l \leq D x \leq u)$ is mixed equisatisfiable to $(l \leq D x \leq u)$.\label{corollary:SplittingEquisatisfiability}
\end{corollary}

\section{Experiments}
\label{SE:experiments}
\label{SE:Experiments}

\pgfplotsset{every axis/.append style={very thick}}
\pgfplotsset{every mark/.append style={solid}}
\tikzset{every mark/.append style={scale=10}}
\begin{figure}[t]
\begin{minipage}[t]{\textwidth}
\scalebox{0.8}{

}
\end{minipage}
\caption{horizontal axis: \# of solved instances; vertical axis: time (seconds)}
\label{figplot1n2}
\end{figure}

\begin{figure}[t]
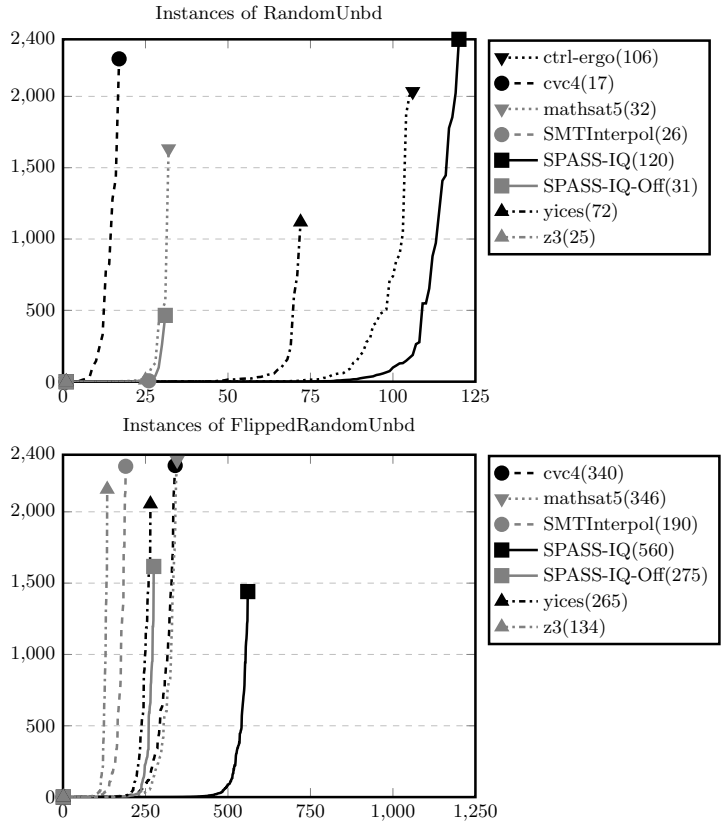

\begin{minipage}[t]{\textwidth}
\scalebox{0.8}{
%
}
\end{minipage}
\caption{horizontal axis: \# of solved instances; vertical axis: time (seconds)}
\label{figplot3n4}
\end{figure}

We integrated the Double-Bounded reduction and the Mixed-Echelon-Hermite transformation into our own theory solver 
\emph{SPASS-IQ v0.2}\footnote{\label{footnote:webpage}Available on \url{http://www.spass-prover.org/spass-iq}} and ran it on four families of newly constructed benchmarks\footnotemark[6].
Once with the transformations turned on (\emph{SPASS-IQ}) and once with the transformations turned off (\emph{SPASS-IQ-Off}). 
If SPASS-IQ encounters a system $A x \leq b$ that is not explicitly bounded, i.e., where not all variables have an explicit upper and lower bound, 
then it computes an equality basis for $A x \leq 0^m$. 
This basis is used to determine whether the system is implicitly bounded, absolutely unbounded or partially bounded, as well as which of the inequalities are bounded. 
Our solver only applies our two transformations if the problem is partially unbounded. 
The resulting equisatisfiable but bounded problem is then solved via branch-and-bound. 
The other two cases, absolutely unbounded and implicitly bounded, are solved respectively via the unit cube test~\cite{BrombergerWeidenbach:17} and branch-and-bound on the original system. 
Our solver also converts any mixed solutions from the transformed system into mixed solutions for the original system following the proof of Lemma~\ref{lemma:mixedsoundness}. 
Rational conflicts are converted between the two systems by using Corollary~\ref{corollary:MCTCERT}.

We tried to restrict our benchmarks to partially unbounded problems since we only apply our transformations on those problems. 
We even found some partially unbounded problems in the SMT-LIB benchmarks for QF\_LIA (quantifier free linear arithmetic). 
However, there are not many such benchmarks: only one in \textit{CAV-2009}, five in \textit{cut\_lemmas}, and three in \textit{slacks}. 
So we created in addition four new benchmark families:

\textit{SlackedQFLIA}: are linear integer benchmarks based on the SMT-LIB classes \textit{CAV-2009}~\cite{Dillig:09}, \textit{cut\_lemmas}~\cite{Griggio:12}, and \textit{dillig}~\cite{Dillig:09}. 
We simply took all of the unsatisfiable benchmarks and replaced in them all variables $x$ with $x_+ - x_-$ where $x_+$ and $x_-$ are two new variables such that $x_+, x_- \geq 0$. 
This transformation, called slacking, 
is equisatisfiable and the slacked version of the \textit{dillig}-benchmarks, 
called \textit{slacked}~\cite{JovanovicdeMoura:13}, 
is already in the SMT-LIB. 
Slacking turns any unsatisfiable problem into a partially unbounded one. 
Hence, all problems in \textit{SlackedQFLIA} are partially unbounded. 
Slacking is commonly used to integrate absolute values into linear systems or 
for solvers that require non-negative variables~\cite{Schrijver:86}.

\textit{RandomUnbd}: are linear integer benchmarks that are all partially unbounded and satisfiable with 10, 25, 50, 75, and 100 variables. 
All problems are randomly created via a sagemath script\footnotemark[6].

\textit{FlippedQFLIA} and \textit{FlippedRandomUnbd}: 
are linear mixed benchmarks that are all partially unbounded. 
They are based on \textit{SlackedQFLIA} and \textit{RandomUnbd}. 
We constructed them by first copying ten versions of the integer benchmarks and then randomly flipping the type of some of the variables to rational (probability of 20\%). 
Some of the flipped instances of \textit{SlackedQFLIA} became satisfiable.

We compared our solver with some of the state-of-the-art SMT solvers currently available for
linear mixed arithmetic: \emph{cvc4-1.5}~\cite{BarrettCDHJKRT:11}, \emph{mathsat5-5.1}~\cite{CimattiGriggio:13}, \emph{SMTInterpol 2.1-335-g4c543a5}~\cite{ChristHoenickeNutz2012}, \emph{yices2.5.4}~\cite{Dutertre:14}, and \emph{z3-4.6.0}~\cite{deMouraBjorner:08}.
Most of these solvers employ a branch-and-bound approach with an underlying dual simplex solver~\cite{DutertredeMoura:06}, 
which is also the basis for our own solver. 
As far as we are aware, none of them employ any techniques that guarantee termination. 

SMTInterpol extends branch-and-bound via the cuts from proofs approach, 
which uses the Mixed-Echelon-Hermite transformation to find more versatile branches and cuts~\cite{ChristHoenicke2015}. 
Although the procedure is not complete, the similarities to our own approach make an interesting comparison. 
Actually, the Double-Bounded reduction alone would be sufficient to make SMTInterpol terminating 
since it already builds branches via a Mixed-Echelon-Hermite transformation. 
 
We also compared our solver with the \emph{ctrl-ergo} solver~\cite{BobotCCIMMM:12} although it is restricted to pure integer arithmetic. 
Ctrl-ergo is complete over linear integer arithmetic and 
uses the most similar approach to our transformations that we found in the literature. 
It dynamically eliminates one linear independent bounded direction at a time via transformation.  
The disadvantages of the dynamic approach are that 
it is very restrictive and does not leave enough freedom to change strategies or to add complementing techniques.
Moreover, ctrl-ergo uses this transformation approach for all problems and not only the partially unbounded ones, 
which sometimes leads to a massive overhead on bounded problems.

For the experiments, we used a Debian Linux cluster and allotted to each problem and solver combination 2 cores of an Intel Xeon E5620 (2.4 GHz) processor, 4 GB RAM, and 40 minutes. 
The only solver benefiting from multiple cores is SMTInterpol. 
The plots in Figures~\ref{figplot1n2} and~\ref{figplot3n4} depict the results of the different solvers. 
In the legends of the plots, the numbers behind the solver names are the number of solved instances. 
For \textit{FlippedQFLIA}, there are two numbers to indicate the number of satisfiable/unsatisfiable instances solved. 
This is only necessary for \textit{FlippedQFLIA} because it is the only tested benchmark family with satisfiable and unsatisfiable instances.
(We verified that the results match if two solvers solved the same problem.
\footnote{The only discrepancies occurred with mathsat and yices. Both  return on different problems satisfiable, where no other solver returns satisfiable but multiple solvers unsatisfiable. We checked their returned ``satisfiable'' assignments and they were in fact not correct. We contacted the mathsat and yices teams about the bugs and excluded the supposedly wrong results from the experiments table.}) 

Although our solver could not solve all problems (due to time and memory limits) it was still able to solve more problems than the other solvers. 
It was also faster on most instances than the other solvers. 
In some of the unsatisfiable, partially unbounded benchmarks ctrl-ergo is better than SPASS-IQ. 
This is due to its conflict focused, dynamic approach. 
For the same reason, ctrl-ergo is slower on the satisfiable, partially unbounded benchmarks.  
Only SPASS-IQ, ctrl-ergo, and yices solved all of the ten original SMT-LIB benchmarks that are partially unbounded, though the complete methods were still a lot faster (SPASS-IQ took 23s, ctrl-ergo took 42s, and yices took 1273s). 
On one of these benchmarks, 20-14.slacks.smt2 from \textit{slacks}, all other solvers seem to diverge. 
Another interesting result of our experiments is that relaxing some integer variables to rational variables seems to make the problems harder instead of easier. 
We expected this for our transformations because the resulting systems become more complex and less sparse, 
but it is also true for the other solvers. 
The reason might be that bound refinement, a technique used in most branch-and-bound implementations, 
is less effective on mixed problems.

The time SPASS-IQ needs to detect the bounded inequalities and to apply our transformations is negligible. 
This is even true for the implicitly bounded problems we tested. 
As mentioned before, we do not have to apply our transformations to terminate on bounded problems. 
This is also the only advantage we gain from detecting that a problem is implicitly bounded.
Since there is no noticeable difference in the run time, we do not further elaborate the results on bounded problems, e.g. with graphs.

An actual disadvantage of our approach is that the Mixed-Echelon-Hermite transformation increases the density of 
the coefficient matrix as well as the absolute size of the coefficients. 
Both are important factors for the efficiency of the underlying simplex solver. 
Moreover, SPASS-IQ reaches more often the memory limit than the time limit because it needs a (too) large number of branches and bound refinements before terminating.

\section{Conclusion}
\label{SE:Conclusion}

We have presented the Mixed-Echelon-Hermite transformation (Lemma~\ref{lemma:MEHNF}) and the Double-Bounded reduction (Lemma~\ref{lemma:splitequivalence} \& Corollary~\ref{corollary:SplittingEquisatisfiability}). 
We have shown that both transformations together turn any constraint system into an equisatisfiable system that is also bounded (Lemma~\ref{lemma:LTDBS}). 
This is sufficient to make branch-and-bound, and many other linear mixed decision procedures, complete and terminating. 
We have also shown how to convert certificates of (un)satisfiability efficiently between the transformed and original systems (Corollary~\ref{corollary:MCTCERT} \& Lemma~\ref{lemma:mixedsoundness}).
Moreover, experimental results on partially unbounded benchmarks show that our approach is also efficient in practice.

Our approach can be nicely combined with the extensive branch-and-bound framework and its many extensions, 
where other complete techniques cannot be used in a modular way~\cite{BobotCCIMMM:12,BrombergerSturmWeidenbach:15}. For future research, we plan to test our transformations in combination with other algorithms, e.g., cuts from proofs, or as a dynamic version similar to the approach used by ctrl-ergo~\cite{BobotCCIMMM:12}.
We also want to test whether our transformations are useful preprocessing steps for select constraint system classes that are bounded.

\bibliographystyle{abbrv}
\bibliography{paperbib}

\begin{appendix}

\section{Incremental Implementation}
\label{SE:incremental}
\label{SE:Incremental Implementation}

Suppose an SMT theory solver has to solve $(A x \leq b) \cup (D x \leq c)$. 
Moreover, the last problem it has solved was $(A x \leq b)$. 
Then we call the run time advantage it gains from having already solved a subset of the problem its \emph{incremental efficiency}.
Since all problems sent to an SMT theory solver are incrementally connected, 
its incremental efficiency is a major factor in determining its total efficiency.

In this Section, we explain 
how to implement our transformation based approach more incrementally efficient and 
what limits there are with regard to incremental efficiency. 
We start our discussion with incrementally efficient implementations of the subcomponents of our procedure: finding bounded inquealities and computing the Mixed-Echelon-Hermite normal form.

\subsection{Finding Bounded Inequalities Incrementally}
\label{SSE:boundedbasis}

In Section~\ref{SE:splitting}, we explained via Lemma~\ref{lemma:boundsandequalities} that 
all bounded directions in $A x \leq b$ are equalities in $A x \leq 0^m$ and vice versa. 
Based on this fact, we recommend to use the method outlined in~\cite{BrombergerWeidenbach:16a} to compute an equality basis for $A x \leq 0^m$ 
and to determine with it the inequalities in $A x \leq b$ that are bounded. 

We also recommend the equality basis method because it is incrementally efficient~\cite{BrombergerWeidenbach:17}. 
This incremental efficiency directly translates to determining bounded inequalities. 
Since determining the bounded inequalities is the bottleneck of splitting (Section~\ref{SE:splitting}), 
the incremental efficiency also translates to splitting a constraint system.

\subsection{Extending the Mixed-Echelon-Hermite Normal Form Incrementally}
\label{SSE:incrementalmehnf}

\begin{figure}[t]
\begin{minipage}{\linewidth}
\begin{algorithm}[H]
    \algoIfrule
    \caption{$\ExtendMEH(H^{(k)} y \leq u^{(k)}, V^{(k)}, a_{k+1}^T x \leq b_{k+1})$}
    \Input{A system of inequalities $H^{(k)} y \leq u^{(k)}$, where $H^{(k)} \in \mathbb{Q}^{k \times n}$ is in Mixed-Echelon-Hermite normal form and $u^{(k)} \in \mathbb{Q}^{k}$, a mixed column transformation matrix $V^{(k)} \in \mathbb{Q}^{n \times n}$, and an inequality $a_{k+1}^T x \leq b_{k+1}$, where $a_{k+1} \in \mathbb{Q}^{n}$ and $b_{k+1} \in \mathbb{Q}$}
    \Effect{Extends $H^{(k)} y \leq u^{(k)}$ by $a_{k+1}^T V^{(k)} y \leq b_{k+1}$ and transforms it into Mixed-Echelon-Hermite normal form $H^{(k+1)} y \leq u^{(k+1)}$ via mixed column transformation operations, which are stored in $V^{(k+1)}$. $H^{(k+1)}$ has at most one more non-zero column than $H^{(k)}$ and $(H^{(k)} y \leq u^{(k)}) \subset (H^{(k+1)} y \leq u^{(k+1)})$.}
    \Output{($H^{(k+1)} y \leq u^{(k+1)}$,$V^{(k+1)}$)\\}
    $h_{k+1}^T := a_{k+1}^T V^{(k)}$;\label{algline:ExtendMEHIneqTransf}\\
    $p := \RPiv(k,H^{(k)})$;\\
    $j := \RPiv(1,h_{k+1}^T)$;\\
    \lIf{$j > p$}{\Return $\ExtendRat(H^{(k)} y \leq u^{(k)}, V^{(k)}, h_{k+1}^T x \leq b_{k+1},j)$}
    $p := \IPiv(k,H^{(k)})$;\\
    $j := \IPiv(1,h_{k+1}^T)$;\\ 
    \lIf{$j > p$}{\Return $\ExtendInt(H^{(k)} y \leq u^{(k)}, V^{(k)}, h_{k+1}^T x \leq b_{k+1},j)$}
    \Return ($(H^{(k)} y \leq u^{(k)}) \cup (h_{k+1}^T y \leq b_{k+1})$,$V^{(k+1)}$);\label{algline:ExtendMEHNoTransf}
\end{algorithm}
\end{minipage}
\caption{$\ExtendMEH()$ extends a Mixed-Echelon-Hermite normal form by one inequality. All algorithms used in the transformation are based on algorithms from~\cite{Schrijver:86}.}
\label{alg:ExtendMEH}
\end{figure}

\begin{figure}[t]
\begin{minipage}{\linewidth}
\begin{algorithm}[H]
    \algoIfrule
    \caption{$\RPiv(k, H)$}
    \Input{A row dimension $k$ and a matrix $H \in \mathbb{Q}^{k \times n}$\\}
    \Return If all columns $j \leq n_1$ in $H$ are $0^k$, then $0$ is returned. Otherwise, this function returns the largest $0 < j \leq n_1$ such that column $j$ in $H$ is not $0^k$. 
\end{algorithm}
\end{minipage}
\begin{minipage}{\linewidth}
\begin{algorithm}[H]
    \algoIfrule
    \caption{$\IPiv(k, H)$}
    \Input{A row dimension $k$ and a matrix $H \in \mathbb{Q}^{k \times n}$\\}
    \Return If all columns $j > n_1$ in $H$ are $0^k$, then $n_1$ is returned instead. Otherwise, this function returns the largest $n_1 < j \leq  n $ such that column $j$ in $H$ is not $0^k$. 
\end{algorithm}
\end{minipage}
\caption{Helper functions}
\label{alg:rpivandipiv}
\end{figure}

\begin{figure}[t]
\begin{minipage}{\linewidth}
\begin{algorithm}[H]
    \algoIfrule
    \caption{$\ExtendRat(H^{(k)} y \leq u^{(k)}, V^{(k)}, h_{k+1}^T y \leq b_{k+1},j)$}
    \Input{A system of inequalities $H^{(k)} y \leq u^{(k)}$, where $H^{(k)} \in \mathbb{Q}^{k \times n}$ is in Mixed-Echelon-Hermite normal form and $u^{(k)} \in \mathbb{Q}^{k}$, a mixed column transformation matrix $V^{(k)} \in \mathbb{Q}^{n \times n}$, an inequality $h_{k+1}^T y \leq b_{k+1}$, where $h_{k+1} \in \mathbb{Q}^{n}$ and $b_{k+1} \in \mathbb{Q}$, and a column index $j \leq n_1$ such that $h_{k+1j} \neq 0$ and $h^{(k)}_{ij} = 0$ for all $i \leq k$}
    \Effect{Extends $H^{(k)} y \leq u^{(k)}$ by $h_{k+1}^T y \leq b_{k+1}$ and transforms it into Mixed-Echelon-Hermite normal form $H^{(k+1)} y \leq u^{(k+1)}$ via mixed column transformation operations, which are stored in $V^{(k+1)}$. $H^{(k+1)}$ has one more non-zero rational column than $H^{(k)}$ and $(H^{(k)} y \leq b^{(k)}) \subset (H^{(k+1)} y \leq u^{(k+1)})$.}
    \Output{($H^{(k+1)} y \leq u^{(k+1)}$,$V^{(k+1)}$)\\}
    $V^{(k+1)} := V^{(k)}$;\\
    $p := \RPiv(k,H^{(k)})$ + 1;\\
    $H^{(k+1)} y \leq u^{(k+1)} := $ Insert $h_{k+1}^T y \leq b_{k+1}$ between rows $p-1$ and $p$ of $H^{(k)} y \leq u^{(k)}$;\label{algline:ExtendRatInsert}\\
    $V^{(k+1)} := $ Swap Column $j$ and $p$ in $V^{(k+1)}$;\\
    $H^{(k+1)} := $ Swap Column $j$ and $p$ in $H^{(k+1)}$;\\
    $V^{(k+1)} := $ Divide Column $p$ of $V^{(k+1)}$ by $h^{(k+1)}_{pp}$;\\
    $H^{(k+1)} := $ Divide Column $p$ of $H^{(k+1)}$ by $h^{(k+1)}_{pp}$;\\
    \For{$j \in \{1, \ldots, p-1, p+1, \ldots, n\}$}{
      $V_{k+1} := $ Subtract $h^{(k+1)}_{pj}$ times column $p$ from column $j$ of $V_{k+1}$;\\
      $H^{(k+1)} := $ Subtract $h^{(k+1)}_{pj}$ times column $p$ from column $j$ of $H^{(k+1)}$;\\
    }
    \Return ($H^{(k+1)} y \leq u^{(k+1)}$,$V^{(k+1)}$);
\end{algorithm}
\end{minipage}
\caption{$\ExtendRat(H^{(k)} y \leq u^{(k)}, V^{(k)}, h_{k+1}^T y \leq b_{k+1},j)$}
\label{alg:ExtendRat}
\end{figure}

\begin{figure}[t]
\begin{minipage}{\linewidth}
\begin{algorithm}[H]
    \algoIfrule
    \caption{$\ExtendInt(H^{(k)} y \leq u^{(k)}, V^{(k)}, h_{k+1}^T y \leq b_{k+1},j)$}
    \Input{A system of inequalities $H^{(k)} y \leq u^{(k)}$, where $H^{(k)} \in \mathbb{Q}^{k \times n}$ is in Mixed-Echelon-Hermite normal form and $u^{(k)} \in \mathbb{Q}^{k}$, a mixed column transformation matrix $V^{(k)} \in \mathbb{Q}^{n \times n}$, an inequality $h_{k+1}^T y \leq b_{k+1}$, where $h_{k+1} \in \mathbb{Q}^{n}$ and $b_{k+1} \in \mathbb{Q}$, and a column index $j > n_1$ such that $h_{k+1j} \neq 0$ and $h^{(k)}_{ij} = 0$ for all $i \leq k$}
    \Effect{Extends $H^{(k)} y \leq u^{(k)}$ by $h_{k+1}^T y \leq b_{k+1}$ and transforms it into Mixed-Echelon-Hermite normal form $H^{(k+1)} y \leq b^{(k+1)}$ via mixed column transformation operations, which are stored in $V^{(k+1)}$. $H^{(k+1)}$ has one more non-zero integer column than $H^{(k)}$ and $(H^{(k)} y \leq u^{(k)}) \subset (H^{(k+1)} y \leq u^{(k+1)})$.}
    \Output{($H^{(k+1)} y \leq u^{(k+1)}$,$V^{(k+1)}$)\\}
    $p := \IPiv(k,H^{(k)})$ + 1;\\
    $H^{(k+1)} y \leq u^{(k+1)} := $ Insert $h_{k+1}^T y \leq b_{k+1}$ between rows $p-1$ and $p$ of $H^{(k)} y \leq u^{(k)}$;\label{algline:ExtendIntInsert}\\
    $(H^{(k+1)}, V^{(k+1)}) := \ReduceLeftInt(H^{(k+1)}, V^{(k)},p);$\\  
    $(H^{(k+1)}, V^{(k+1)}) := \ReduceRightInt(H^{(k+1)}, V^{(k+1)},p);$\\
    \Return ($H^{(k+1)} y \leq u^{(k+1)}$,$V^{(k+1)}$);
\end{algorithm}
\end{minipage}
\caption{$\ExtendInt(H^{(k)} y \leq u^{(k)}, V^{(k)}, h_{k+1}^T y \leq b_{k+1},j)$}
\label{alg:ExtendInt}
\end{figure}

\begin{figure}[t]
\begin{minipage}{\linewidth}
\begin{algorithm}[H]
    \algoIfrule
    \caption{$\ReduceLeftInt(H^{(k+1)}, V^{(k)},p)$}
    \Input{A matrix $H^{(k+1)} \in \mathbb{Q}^{k+1 \times n}$, a mixed column transformation matrix $V^{(k)} \in \mathbb{Q}^{n \times n}$, and a row and column index $p$.}
    \Effect{Applies mixed column transformations to $H^{(k+1)}$ until all entries $h^{(k+1)}_{pi}$ with $i > p$ are zero. The transformations are combined with the previous transformations into $V^{(k+1)}$. The overall algorithm is based on the Euclidean algorithm for GCD computation.}
    \Output{($H^{(k+1)}$,$V^{(k+1)}$)\\}
    $V^{(k+1)} := V^{(k)};$\\
    \tcc{Since this algorithm is based on GCD computation, we need to abstract the coefficients $h^{(k+1)}_{pi}$ to integers. (Stored in $S$.)}
	$(H^{(k+1)},V^{(k+1)},S) := \AbstractToInt(H^{(k+1)},V^{(k+1)},p)$;\\
    
    \tcc{Next we perform the Euclidean algorithm via column operations on the coefficients stored in $S$.}
    \While{$|S| \neq 1$}{
      $(i,s_{pi}) := $ an $(j,s_{pj}) \in S$ with the smallest $s_{pj}$;\\
      \For{$(j,s_{pj}) \in S$}{
         \lIf{$j = i$}{
         	continue
         }
         $S := S \setminus \{(j,s_{pj})\}$;\\
         $d_{pj} := \floor(s_{pj} \div s_{pi})$;\\
         $s_{pj} := s_{pj} - d_{pj} \cdot s_{pi}$;\\
         $V^{(k+1)} := $ Subtract $d_{pj}$ times column $i$ from column $j$ of $V^{(k+1)}$;\\
         $H^{(k+1)} := $ Subtract $d_{pj}$ times column $i$ from column $j$ of $H^{(k+1)}$;\\
         \lIf{$s_{pj} \neq 0$}{
            $S := S \cup \{(j,s_{pj})\}$
         }
      }
    }
    \tcc{We have found the gcd $s_{pi}$ as soon as $S$ contains only one element $(i,s_{pi})$. We swap it to column $p$.}
    $(i,s_{pi}) := $ the only $(i,s_{pi}) \in S$;\\
    $V^{(k+1)} := $ Swap Column $i$ and $p$ in $V^{(k+1)}$;\\
    $H^{(k+1)} := $ Swap Column $i$ and $p$ in $H^{(k+1)}$;\\
    
    \Return ($H^{(k+1)}$,$V^{(k+1)}$);
\end{algorithm}
\end{minipage}
\caption{$\ReduceLeftInt(H^{(k+1)}, V^{(k)},p)$}
\label{alg:ReduceLeftInt}
\end{figure}

\begin{figure}[t]
\begin{minipage}{\linewidth}
\begin{algorithm}[H]
    \algoIfrule
    \caption{$\AbstractToInt(H^{(k+1)},V^{(k+1)},p)$}
    \Input{A matrix $H^{(k+1)} \in \mathbb{Q}^{k+1 \times n}$, a mixed column transformation matrix $V^{(k+1)} \in \mathbb{Q}^{n \times n}$, and a row and column index $p$.}
    \Effect{Negate all columns $i \geq p$ with $h^{(k+1)}_{pi} < 0$. 
            Extract the integer part $s_{pi}$ of each coefficient $h^{(k+1)}_{pi}$. 
            Store all non-zero integer parts $s_{pi}$ and their column index $i$ in a set $S$.}
    \Output{($H^{(k+1)}$,$V^{(k+1)}$,$S$)\\}
    $S := \emptyset$;\\
    $c := \lcm\{d_{pj} \; \mid \; j \in \{n_1+1, \ldots, n\} \mbox{ and } d_{pj} := \mbox{ denominator of } h^{(k+1)}_{pj}\}$;\\
    \For{$j \in \{p, \ldots, n\}$}{
      \If{$h^{(k+1)}_{pj} < 0$}{
         $V^{(k+1)} := $ Negate column $i$ of $V^{(k+1)}$;\\
         $H^{(k+1)} := $ Negate column $i$ of $H^{(k+1)}$;\\}
      \If{$h^{(k+1)}_{pj} > 0$}{
         $S := S \cup \{(j,h^{(k+1)}_{pj} \cdot c)\}$;\\}
    }
    \Return ($H^{(k+1)}$,$V^{(k+1)}$,$S$);
\end{algorithm}
\end{minipage}
\caption{$\AbstractToInt(H^{(k+1)},V^{(k+1)},p)$}
\label{alg:ReduceRightInt}
\end{figure}

\begin{figure}[t]
\begin{minipage}{\linewidth}
\begin{algorithm}[H]
    \algoIfrule
    \caption{$\ReduceRightInt(H^{(k+1)}, V^{(k+1)},p)$}
    \Input{A matrix $H^{(k+1)} \in \mathbb{Q}^{k+1 \times n}$, a mixed column transformation matrix $V^{(k+1)} \in \mathbb{Q}^{n \times n}$, and a row and column index $p$.}
    \Effect{Applies mixed column transformations to $H^{(k+1)}$ until all entries $h^{(k+1)}_{pi}$ with $n_1 < i < p$ are non-negative and less than $h^{(k+1)}_{pp}$. The transformations are also added to $V^{(k+1)}$.}
    \Output{($H^{(k+1)} y \leq b^{(k+1)}$,$V^{(k+1)}$)\\}
    $c := \lcm\{d_{pj} \; \mid \; j \in \{n_1+1, \ldots, n\} \mbox{ and } d_{pj} := \mbox{ denominator of } h^{(k+1)}_{pj}\}$;\\
    $s_{pp} := h^{(k+1)}_{pp} \cdot c$;\\
    \For{$j \in \{n_1 + 1, \ldots, p-1\}$}{
      $s_{pj} := h^{(k+1)}_{pj} \cdot c$;\\
      $d_{pj} := \floor(s_{pj} \div s_{pp})$;\\
      $V^{(k+1)} := $ Subtract $d_{pj}$ times column $p$ from column $j$ of $V^{(k+1)}$;\\
      $H^{(k+1)} := $ Subtract $d_{pj}$ times column $p$ from column $j$ of $H^{(k+1)}$;\\
    }
    \Return ($H^{(k+1)}$,$V^{(k+1)}$);
\end{algorithm}
\end{minipage}
\caption{$\ReduceRightInt(H^{(k+1)}, V^{(k+1)},p)$}
\label{alg:ReduceRightInt}
\end{figure}

The Mixed-Echelon-Hermite normal form (MEHNF) can also be computed incrementally efficient. 
However, the polynomial time algorithms for computing the reduced echelon column form and the hermite normal form are 
typically less incrementally efficient than the ones based on Gaussian elimination. 
So to achieve incremental efficiency, we have to accept a worst case exponential run time. 
Fortunately, the Gaussian based transformations rarely seem to reach their exponential worst case in practice. 
Our experiments support this assumption since the transformation cost is negligible (if not immeasurable) on all of the tested benchmarks.

Before we can explain the incrementally efficient version of the Mixed-Echelon-Hermite transformation 
(see $\ExtendMEH()$ in Figure~\ref{alg:ExtendMEH})), 
we need to introduce one final notation: we denote by $A^{(k)}$ the $k$-th element of a series of matrices. 

With the algorithm $\ExtendMEH()$, 
we incrementally compute the MEHNFs $H^{(k)} y \leq u^{(k)}$ for the constraint systems $A^{(k)} x \leq b^{(k)}$, 
where $A^{(k)} := (a_1, \ldots, a_k)^T$ and $b^{(k)} := (b_1, \ldots, b_k)^T$. 
As already mentioned in Section~\ref{SE:doublebounded}, it is not possible to transform every matrix $A^{(k)} \in \mathbb{Q}^{k \times n}$ into Mixed-Echelon-Hermite normal form. 
We have to restrict ourselves to matrices, where the upper left $r \times n_1$ submatrix has the same rank $r$ as the complete left $k \times n_1$ submatrix. 
This is very easy to accomplish because we are looking at constraint systems $A^{(k)} x \leq b^{(k)}$ and not just matrices. 
This means we can simply swap the inequalities in $A^{(k)} x \leq b^{(k)}$ to get the systems $C^{(k)} x \leq u^{(k)}$, 
where $C^{(k)}$'s upper left $r \times n_1$ submatrix has the desired form. 
(Note that this is explicitly done by $\ExtendMEH()$.) 
So $H^{(k)} := C^{(k)} V^{(k)}$ and not $H^{(k)} := A^{(k)} V^{(k)}$ for an appropriate mixed transformation matrix $V^{(k)}$.

$\ExtendMEH()$ works as follows:
Initially, our MEHNF $H^{(0)} y \leq u^{(0)}$ is just the empty set and our transformation matrix $V^{(0)}$ is just the $n \times n$ identity matrix. 
Then we incrementally extend $H^{(k)} y \leq u^{(k)}$ and $V^{(k)}$ one inequality $a_{k+1}^T x \leq b_{k+1}$ at a time by computing $(H^{(k+1)} y \leq u^{(k+1)},V^{(k+1)}) := \ExtendMEH(H^{(k)} y \leq u^{(k)}, V^{(k)}, a_{k+1}^T x \leq b_{k+1})$. 
(Note that $V^{(k+1)}$ encompasses all column transformations necessary to transform $C^{(k+1)}$ into $H^{(k+1)}$.)

To this end, $\ExtendMEH()$ first applies the previous column transformations $V^{(k)}$ to $a_{k+1}^T x \leq b_{k+1}$ 
to get the inequality $h_{k+1}^T y \leq b_{k+1}$ (line~\ref{algline:ExtendMEHIneqTransf}). 
Next, $\ExtendMEH()$ checks whether $h_{k+1}^T$ has any non-zero entries $h_{k+1j}$ in one of the zero columns $j$ of $H^{(k)}$. 
If $h_{k+1}^T$ does not have any such entries, 
then no column transformations are necessary and $H^{(k+1)} y \leq u^{(k+1)} := (H^{(k)} y \leq u^{(k)}) \cup (h_{k+1}^T y \leq b_{k+1})$ is in MEHNF (line~\ref{algline:ExtendMEHNoTransf}). 
Otherwise, $(H^{(k)} y \leq u^{(k)}) \cup (h_{k+1}^T y \leq b_{k+1})$ is not in MEHNF 
because $h_{k+1}^T$ fills one of the gaps of $H^{(k)}$, i.e, has a non-zero coefficient in a zero column of $H^{(k)}$. 
In order to resolve this, we have to distinguish between two cases:

Case 1: $h_{k+1}^T$ fills a rational gap of $H^{(k)}$, i.e., 
there exists a zero column $0 < j \leq n_1$ in $H^{(k)}$ such that $h_{k+1j} \neq 0$. 
In this case, we have to extend $H^{(k)}$ from $p-1$ non-zero rational columns to $p$ non-zero rational columns. 
We do so with the function $\ExtendRat()$ (Figure~\ref{alg:ExtendRat}).
$\ExtendRat()$ first inserts the inequality $h_{k+1}^T y \leq b_{k+1}$ at an appropriate position $p$ (line~\ref{algline:ExtendRatInsert}), to solve the rank requirements we discussed before. 
So in the new constraint system $(H^{(k+1)} y \leq u^{(k+1)})$ the inequality $h_{k+1}^T y \leq b_{k+1}$ is located in row $p$. 
Then $\ExtendRat()$ swaps column $j$ with column $p$ and 
uses column operations to eliminate all other coefficients in $h_{k+1}^T$ that have filled gaps in $H^{(k)}$. 
The result $H^{(k+1)} y \leq u^{(k+1)}$ is then again in MEHNF (and $V^{(k+1)}$ is the transformation matrix as specified above).
Since all column operations are performed on columns with gaps in $H^{(k)}$, all inequalities in $H^{(k)} y \leq u^{(k)}$ also appear in $H^{(k+1)} y \leq u^{(k+1)}$, i.e., $(H^{(k)} y \leq u^{(k)}) \subset (H^{(k+1)} y \leq u^{(k+1)})$.

Case 2: $h_{k+1}^T$ fills no rational gap, but an integer gap of $H^{(k)}$, i.e., 
there exists a zero column $n_1 < j \leq n$ in $H^{(k)}$ such that $h_{k+1j} \neq 0$. 
In this case, we have to extend $H^{(k)}$ from $p-1$ non-zero integer columns to $p$ non-zero integer columns. 
We do so with the function $\ExtendInt()$ (Figure~\ref{alg:ExtendInt}).
$\ExtendInt()$ first inserts the inequality $h_{k+1}^T y \leq b_{k+1}$ at an appropriate position $p$ (line~\ref{algline:ExtendIntInsert}), to solve the rank requirements we discussed before. 
So in the new constraint system $(H^{(k+1)} y \leq u^{(k+1)})$ the inequality $h_{k+1}^T y \leq b_{k+1}$ is located in row $p$. 
Then $\ExtendInt()$ swaps column $j$ with column $p$ and 
uses column operations to eliminate all other coefficients in $h_{k+1}^T$ that have filled gaps in $H^{(k)}$. 
The result $H^{(k+1)} y \leq u^{(k+1)}$ is then again in MEHNF (and $V^{(k+1)}$ is the transformation matrix as specified above).
Since all column operations are performed on columns with gaps in $H^{(k)}$, all inequalities in $H^{(k)} y \leq u^{(k)}$ also appear in $H^{(k+1)} y \leq u^{(k+1)}$, i.e., $(H^{(k)} y \leq u^{(k)}) \subset (H^{(k+1)} y \leq u^{(k+1)})$.

The case distinction over the algorithms $\ExtendRat()$ and $\ExtendInt()$ is necessary 
because of the restrictions we have on our column transformations\footnote{Without these restrictions, our transformations would not be mixed equisatisfiable!}, e.g., 
we can add multiples of rational columns to integer columns but not vice versa. 

Since $\ExtendRat()$ and $\ExtendInt()$ change only the new inequality, 
it holds that $H^{(i)} := C^{(i)}V^{(k)}$ for all $i \leq k$. 
This means that an extended transformation matrix still transforms the previous constraint systems into MEHNF. 
We can use this fact to also make backtracking\footnote{Removing inequalities in the order they were added.} efficient. 
In order to remove $a_{k}^T x \leq b_{k}$ from $H^{(k)} y \leq u^{(k)}$, 
we simply remove the $k$-th inequality that was added to the constraint system (can be efficiently marked with a flag) to get again $(H^{(k-1)} y \leq u^{(k-1)})$.
Since $H^{(i)} := C^{(i)}V^{(k)}$ for all $i \leq k$, 
it is not necessary to change the transformation matrix\footnote{When the size of coefficients in $V^{(k)}$ gets too large, 
it can make sense to recompute $H^{(k-1)}$ and $V^{(k-1)}$ to get a smaller transformation matrix.}.
Thus, we have found an incrementally and decrementally efficient way to compute the MEHNF of a constraint system.

\subsection{The Complete Incremental Procedure}

Now that we have incrementally efficient subprocedures, 
we can describe a version of our complete procedure that is incrementally efficient. 
As a reminder, the non-incremental version of our total procedure works as follows:
Our input is a constraint system $A x \leq b$ and we want to find a mixed solution for it. 
To this end, we first compute the equality basis of $A x \leq 0^m$ to find the inequalities and directions in $A x \leq b$ that are bounded. 
Next we do a case distinction depending on whether $A x \leq b$ is bounded, absolutely unbounded or partially unbounded. 
If $A x \leq b$ is bounded, we find the mixed solution via branch-and-bound~\footnote{We recommend to use the version of the dual simplex solver presented by Dutertre and de Moura~\cite{DutertredeMoura:06} as the basis for the underlying branch-and-bound solver. We do so because this version is highly incrementally efficient.}. 
If $A x \leq b$ is absolutely unbounded, we find the mixed solution via the unit cube test~\cite{BrombergerWeidenbach:17}. 
The only slightly complicated case is if $A x \leq b$ is partially unbounded. 
In this case, we first split $A x \leq b$ into a split system and transform the double-bounded part into its MEHNF. 
The double-bounded system in MEHNF is then solved with branch-and-bound.

Now assume that we have done all of the above for $A x \leq b$, but need to incrementally extend it to $(A x \leq b) \cup (A' x \leq b')$.
This means we want to find a mixed solution for $(A x \leq b) \cup (A' x \leq b')$. 
If $(A x \leq b)$ was already bounded, then we know that $(A x \leq b) \cup (A' x \leq b')$ will also be bounded and we simply apply branch-and-bound to it. 
Otherwise, we have to extend the equality basis of $A x \leq 0^m$ to the equality basis of $(A x \leq 0^m) \cup (A' x \leq 0^{m'})$ and use it to find the inequalities and directions in $(A x \leq b) \cup (A' x \leq b')$ that are bounded. 
In Section~\ref{SSE:boundedbasis}, we have shown how to do this incrementally efficient.  
Next we do a case distinction depending on whether $(A x \leq b) \cup (A' x \leq b')$ is bounded, absolutely unbounded or partially unbounded. 
If $(A x \leq b) \cup (A' x \leq b')$ is now bounded, we find the mixed solution via branch-and-bound. 
If $(A x \leq b) \cup (A' x \leq b')$ is still absolutely unbounded, we find the mixed solution via the unit cube test (also an incrementally efficient procedure)~\cite{BrombergerWeidenbach:17}. 
If $(A x \leq b) \cup (A' x \leq b')$ is still partially unbounded, we continue as follows: 
We still have the split system for $(A x \leq b)$ and can now use our extended equality basis for $(A x \leq 0^m) \cup (A' x \leq 0^{m'})$ to efficiently extend it to a split system for $(A x \leq b) \cup (A' x \leq b')$. 
Since adding new inequalities can only add bounded directions, the double-bounded part of the extended split system still contains all bounded inequalities from the previous double-bounded part.
This means we can incrementally extend the MEHNF $l \leq H y \leq u$ by the new inequalities in the double-bounded part of $(A x \leq b) \cup (A' x \leq b')$. 
In Section~\ref{SSE:incrementalmehnf}, we have shown how to do this incrementally efficient. 
Finally, we solve the extended double-bounded constraint system $(l \leq H y \leq u) \cup (l' \leq H' y \leq u')$ with branch-and-bound. 
Since we only add inequalities to the running constraint system $(l \leq H y \leq u)$, 
we can continue our branch-and-bound search incrementally efficient.

This shows that most parts of our procedure can be implemented incrementally efficient. 
However, there are two limits to the incremental efficiency. 
First of all, we have to store multiple constraint systems in our memory to stay incrementally efficient: 
we need one system to store the current equality basis, so we can later extend it; 
we need one system to store the current MEHNF, so we can later extend it;
we need the current transformation matrix of the MEHNF transformation, so we can later extend it; 
and we need one copy of the MEHNF to perform branch-and-bound on. 
Secondly, we do not know how to make the assignment/solution conversion incrementally efficient, i.e., 
how to convert the mixed solution of the transformed system to a mixed solution of the original system in an incrementally efficient way (see Lemma~\ref{lemma:mixedsoundness} for the non-incremental subprocedure). 
However, this second limitation is in reality not a problem 
because there are ways to avoid the conversion until we know that the complete problem is satisfiable. 
So the conversion is used at most once for each SMT input problem. 
In the next subsection, we will elaborate why this is the case.

\subsection{Avoiding Conversion}

In order to explain why we can avoid the conversion, 
we first have to distinguish the origin of the incrementally connected problems, i.e.,
the origin of the problems sent from the SMT solver to the SMT theory solver. 
There are typically two reasons a theory solver might receive incrementally connected problems from the SMT solver: 

(1) The SMT solver tries to prune some partial models (i.e., conjunctions of literals) that are theory unsatisfiable. 
This case is actually not necessary for a complete SMT solver\footnote{Only incomplete models do not have to be checked. Complete models still need to be checked for theory satisfiability!}. 
It is just a trick to speed-up the boolean search of the SMT solver. 
However, it would already be too expensive for the theory solver to check all partial models. 
Instead, they typically just check partial models when the SAT solver is about to do a decision. 
And even then the check is often just a sound approximation of the complete theory solver 
because the complete check is too expensive for some theories. 
One of those theories is in fact linear (mixed) integer arithmetic. 
For this theory, most SMT solvers check only the rational relaxation of the partial models for theory satisfiability. 
So this source of incrementally connected problems is not relevant to our complete approach.

(2) The SMT solver combines multiple theory solvers via the Nelson-Oppen method. 
As part of the Nelson-Oppen method, 
(2.1) each theory solver has to first determine the satisfiability of their own conjunctions of literals. 
(2.2) Then the theory solvers incrementally send to each other (negated) equalities over constant function symbols 
and test these extended problems for satisfiability. 
(2.3) This continues until they find a complete and satisfiable equivalence class over the constant function symbols. 
All of the above can be done with our transformation scheme without converting the intermediate solutions to the original system. 
However, most SMT solvers rely on the intermediate solutions to the original system to guess 
the (negated) equalities they send in step (2.2).

At a first glance, case (2) seems like it actually needs the solution conversion via Lemma~\ref{lemma:mixedsoundness}. 
However, there is an easy and reasonable way to avoid it. 
Instead of using the intermediate solution to the complete original system, 
we just use the intermediate solution to the double-bounded part of the original system. 
This solution can be efficiently computed with the transformation matrix $V$, i.e., 
$x := V y$ is the solution to the double-bounded part of the original system if $y$ is the solution to the transformed system. 
This is a reasonable approximation for the guesses in (2.2)  
because we know that the unbounded part is irrelevant to the satisfiability of the original system (Corollary~\ref{corollary:SplittingEquisatisfiability}). 

We conclude that our total incremental procedure never has to convert a complete solution more than once. 
So the procedure should be incrementally efficient in practice. 
However, we are unable to test this claim with experiments 
since we only have a working theory solver and not a complete SMT solver with multiple theory solvers.

\end{appendix}

\end{document}